\documentclass[12pt,onecolumn]{article}
\usepackage[margin=2.5cm]{geometry}
\usepackage[utf8]{inputenc}
\usepackage{amsmath,amssymb,amsthm} 
\usepackage{color}
\usepackage[font=small,labelfont=bf]{caption}
\usepackage{wrapfig}
\usepackage{enumitem}
\usepackage{authblk}
\usepackage{graphicx}
\usepackage{hyperref}
\usepackage{empheq}
\usepackage{wrapfig}
\usepackage{ulem}
\def \ve {\varepsilon}
\def \r {\boldsymbol{r}}
\def \v {\boldsymbol{v}}
\def \u {\boldsymbol{u}}
\def \n {\boldsymbol{n}}

\newtheorem{theorem}{Theorem}[section]
\newtheorem{proposition}{Proposition}
\newtheorem{remark}{Remark}
 
\begin{document} 

\title{A kinetic approach to active rods dynamics\\ in confined domains}


\author[1]{Leonid Berlyand}
\author[2]{Pierre-Emmanuel Jabin}
\author[1]{Mykhailo Potomkin}
\author[1,3]{El\.{z}bieta Ratajczyk}

\affil[1]{Department of Mathematics, Pennsylvania State University, University Park, 16802, USA}

\affil[2]{Department of Mathematics, University of Maryland, College Park, MD 20742-3289, USA}

\affil[3]{Department of Mathematics, Faculty of Electrical Engineering and Computer Science, Lublin University of
Technology, Nadbystrzycka 38A, 20-618 Lublin, Poland}

\date{}

\maketitle

\begin{abstract}
The study of active matter consisting of many self-propelled (active) swimmers in an imposed  
flow is important for many applications. Self-propelled swimmers may represent both living and artificial ones such as bacteria and chemically driven bi-metallic nano-particles. In this work we focus on a kinetic description of active matter represented by self-propelled rods swimming in a viscous fluid confined by a wall. It is well-known that walls may significantly affect the trajectories of active rods in contrast to unbounded or periodic containers. Among such effects are accumulation at walls and upstream motion (also known as negative rheotaxis). Our first main result is the rigorous derivation of boundary conditions for the active rods' probability distribution function in the limit of vanishing inertia. Finding such a limit is important due to (i) the fact that in many examples of active matter inertia is negligible, since swimming occurs in a low Reynolds number regime, and (ii) this limit allows us to reduce  the dimension $-$ and so computational complexity $-$ of the kinetic description. For the resulting model, we derive the system in the limit of vanishing translational diffusion which is also typically negligible for active particles. This system allows for tracking separately active particles accumulated at walls and active particles swimming in the bulk of the fluid.          
\end{abstract}

\tableofcontents


\section{Introduction}

Recently, active matter has attracted much attention of the scientific community (see e.g. reviews \cite{Ram2010,MarJoaRamLivProRaoSim2013,ElgWinGom2015}). In general, active matter is defined as a system of many agents moving due to consumption of energy stored in the surrounding environment (e.g., chemical or food) and converting it into mechanical force which is called {\it self-propulsion}. The agents exhibiting self-propulsion are named active, as opposed to passive agents which can move only if an external field is applied. There are a vast number of examples of active matter: from suspensions of bacteria \cite{HerStoGra05,HaiSokAraBer09,HaiAroBerKar10,RyaSokBerAra13,kaiser2014transport,PotTouBerAra2016,LopGacDouAurCle2015}, flocks of birds \cite{MotTad2011,CavGia2014,Pop2016}, and schools of fish \cite{KatTunIoa2011,TunKatIoa2013} to crowds of people \cite{PinVelCal2016} which also meet the definition of active matter since people exhibit self-propulsion (walking). Modeling and further analysis of active matter is of great importance due to the variety of striking phenomena and promising applications (reduction of viscosity, cargo delivery for medical purposes, materials repair, etc.).       
   
In this work, we are interested in modeling the wide class of active matter where agents are rod-shaped microswimmers, i.e., the surrounding environment is a viscous fluid and swimming occurs in the low Reynolds number regime. 
Examples of such microswimmers are bacteria (especially, rod-shaped {\it B. subtilis}) and active bi-metallic micro- and nanorods which swim in a viscous solution with hydrogen peroxide \cite{paxton2004catalytic}. It was observed both theoretically and experimentally for various types of microswimmers that their trajectories are much more complex than in the case of passive swimmers which simply follow streamlines of an external field (the background flow). In particular, the following phenomena were observed in dynamics of active microswimmers: accumulation at walls (bordertaxis) and upstream motion (rheotaxis), see \cite{Rot1963,BerTur1990,RamTulPha1993,FryForBerCum1995,VigForWagTam2002,BerTurBerLau2008,HilKalMcMKos2007,fu2012bacterial,YuaRaiBau2015rheo,PalSacAbrBarHanGroPinCha2015,PotTouBerAra2016,PotKaiBerAra2017,rheotaxis2018baker} and references therein. Throughout the paper we use the term ``active rods" for these rod-shaped active microswimmers.     

There are two common mathematical approaches to describe dynamics of an active rod in a viscous fluid.  The first one is based on force and torque balances for each individual active rod. This approach results in a Langevin equation (or a system of coupled Langevin equations) for unknown location, orientation and velocities, both translational and angular, of the active rod. In the second approach, which is also called a kinetic approach, the main unknown is the probability distribution function of the active rod, and the function satisfies the Fokker-Planck equation.  
These two approaches are directly related mathematically: roughly speaking, right hand sides of equations in the first approach are coefficients of the Fokker-Planck equation in the second approach. The second approach is more preferable if one studies statistical properties of a large number of active rods since it does not require many realizations, unlike the first direct one. 

The focus of this work is on the development of a kinetic approach for active rods swimming in a container restricted by a confinement (a wall). One can formulate how an active rod behaves when it collides with the wall (a collision rule) in the first approach. On the other hand, it is not immediately clear how the collision rule translates into a boundary condition for the Fokker-Planck equation. This is because a collision rule is typically a relation between velocities before and after a collision, whereas the Fokker-Planck equation is usually written in the vanishing inertia (overdamped) limit. Thus, the active rods' velocities are no longer variables of the unknown probability distribution function. The vanishing inertia limit of the Fokker-Planck equation  is relevant for the low Reynolds number regime and important since it allows one to reduce dimension and thus drastically decrease computational complexity and even make the equation amenable for analysis.  

Our first result is the rigorous derivation of this limit and, more importantly, the boundary condition for the overdamped Fokker-Planck equation in this limit. Namely, we show that the limiting probability distribution function, which depends on the active rod's location and orientation, satisfies the no-flux condition on the wall for each given orientation.            
We note that similar boundary conditions were phenomenologically and independently derived in \cite{EzhSai2015} to analyze the distribution of active rods inside an infinite channel.

Next, we consider the case of a small translational diffusion which is negligible in experiments for active particles and equated to zero in corresponding individual based models. By using the boundary layer multi-scale approach, we derive the kinetic system in the limit of vanishing translational diffusion.  The significance of the system is that it describes explicitly the population of active rods accumulated at walls.

The structure of this paper is as follows. 
First, in Section~\ref{sec:main-result} we formulate our two main results: on vanishing inertia and vanishing translational diffusion limits. Details of the results' derivation are relegated to Sections~\ref{sec:zero-inertia} and \ref{sec:nodiffusion}. Next, in Section~\ref{sec:numerics} we present a numerical example in which we compare the derived  limiting kinetic models with Monte Carlo simulations for the corresponding individual based model. Finally, in Section~\ref{sec:concrete} we provide a specific physical model of a self-propelled nano-particle swimming in a viscous flow by presenting both individual based model and the corresponding (pre-limiting) Fokker-Planck equation.


\section{Main results}
\label{sec:main-result}
We start with the Fokker-Planck equation describing random dynamics of  an active rod with inertia:  
\begin{eqnarray}\label{fokkerplanck_mr}
&&\partial_t f_{\ve}+\dfrac{1}{\ve}\v\cdot \nabla_{\r}f_\ve+\dfrac{1}{\ve^2}\nabla_{\v}\cdot\left((\ve \u-\v)f_\ve-D_{\mathrm{tr}}\nabla_{\v}f_\ve\right)+ \nonumber
\\
&&\hspace{86pt}+\dfrac{1}{\ve}\omega \partial_{\varphi}f_\ve+\dfrac{1}{\ve^2}\partial_{\omega}\left((\ve T-\omega)f_\ve-D_{\mathrm{rot}}\partial_\omega f_\ve\right)=0.
\end{eqnarray}
The unknown function $ f_\ve(t,\r,\v,\varphi,\omega)$ is the probability distribution function of the active rod's location ${\r}\in \Omega\subset \mathbb R^2$, translational velocity $\v\in \mathbb R^2$, orientation angle $\varphi \in [-\pi,\pi)$, and angular velocity $\omega\in \mathbb R$. Given functions $\u=\u(\r,\varphi)$ and $T=T(\r,\varphi)$ are smooth in $\Omega \times [-\pi,\pi)$ and $2\pi$-periodic in $\varphi$. Small positive parameter $\varepsilon\ll 1$ measures the effect of inertia on dynamics of an active rod. The unknown function $f_\ve$ is $2\pi$-periodic in $\varphi$.
In addition, the following boundary condition is imposed on  $f_\ve$: 
\begin{equation}\label{bc_fokker-planck_mr}
\v f_\ve\cdot \n =\v' f_\ve' \cdot (-\n),\,\,\,\r\text{ on }\Gamma.
\end{equation} 
Here $\Gamma$ is the boundary of $\Omega$ (a wall of the container), $\n$ is the outward normal, and $ f_\ve'=  f_\ve(t,\r,\v',\varphi,\omega')$ where pairs $(\v,\omega)$ and $(\v',\omega')$ represent translational and angular velocities of the active rod before and after a collision with the wall. The relation between the two pairs of velocities is given by:  
\begin{equation}\label{collision_relation}
\left[\begin{array}{c}\v'\\\omega'\end{array}\right]=\mathcal{C}\left[\begin{array}{c}\v\\ \omega \end{array}\right], \quad \mathcal{C}\in \mathbb R^{3\times 3}, \quad |\det{\mathcal{C}}|=1. 
\end{equation}
The specific form of matrix $\mathcal{C}$ as well as the derivation of the kinetic model \eqref{fokkerplanck_mr}-\eqref{bc_fokker-planck_mr} from the individual dynamics of an active rod are relegated to Section \ref{sec:individual-rod}. 

To simplify notations, denote $\mathcal{X}:=({\r},\varphi)$, $ {\mathcal{V}}:=(\v,\omega)$, $\mathcal{U}:=(\u,T)$, and  
$$\mathcal{D}:=\left[\begin{array}{ccc}D_{\mathrm{tr}}&0&0\\0&D_{\mathrm{tr}}&0\\0&0&D_{\mathrm{rot}}\end{array}\right].
$$
In new notations, the Fokker-Planck equation \eqref{fokkerplanck_mr} is 
\begin{equation}\label{fp_with_u_and_diffusion_0_mr}
\partial_t f_\ve + \dfrac{1}{\varepsilon}\mathcal{V}\cdot \nabla_{\mathcal{X}}f_\ve+
\dfrac{1}{\varepsilon^2}\nabla_{{\mathcal{V}}}\cdot\left((\varepsilon \mathcal{U}-\mathcal{V})f_\ve - \mathcal{D}\nabla_{\mathcal{V}}f_\ve\right) = 0.
\end{equation}

{\it Our first main result} is the reduction, for small $\ve$, of the unknown function $f_\ve(t,\mathcal{X},\mathcal{V})$ depending on a 7-dimensional variable to the unknown function $\rho(t,\mathcal{X})$ depending on a 4-dimensional variable. 
The result is formulated in the following theorem. 

\begin{theorem}\label{thm:main_mr}
	Let  $\rho_\ve$ be defined by
	\begin{equation}
	\rho_\ve(t,\mathcal{X}):=\int_{\mathbb R^3}f_\ve(t,\mathcal{X},\mathcal{V})\,\mathrm{d}\mathcal{V},\nonumber
	\end{equation}
	where $f_\ve$ solves Fokker-Planck equation \eqref{fp_with_u_and_diffusion_0_mr} with boundary condition \eqref{bc_fokker-planck_mr}. Then $\rho_\ve$ converges to $\rho$ in the distributional sense, as $\ve\to 0$, where $\rho$ satisfies the following limiting equation
	\begin{equation}\label{limiting_fp_mr}
	\partial_t \rho +\nabla_\mathcal{X}\cdot (\mathcal{U}\rho) =\nabla_{\mathcal{X}}\cdot\mathcal{D}\nabla_\mathcal{X}\rho
	\end{equation}
	with the boundary condition 
	\begin{equation}\label{no_flux_bc_0_mr}
	D_{\mathrm{tr}}\dfrac{\partial \rho}{\partial \n}=(\u\cdot \n)\rho, \quad {\bf r}\text{ on }\Gamma,\quad -\pi\leq \varphi< \pi. 
	\end{equation}    
	Recall that $\mathcal{U}=(\u,T)$.
\end{theorem}

The proof of Theorem \ref{thm:main_mr} is given in Section \ref{sec:zero-inertia}. This theorem means that elastic collision boundary condition \eqref{bc_fokker-planck_mr} for $f_\ve$ transforms into no-flux boundary condition \eqref{no_flux_bc_0_mr} for $\rho$. The vanishing inertia limit in the Fokker-Planck equation for spherical particles (no $\varphi$) was considered in 
  \cite{GouJabVas2004,MelVas2007,MelVas2008}. The main difference, besides no $\varphi$, is that in \cite{GouJabVas2004,MelVas2007,MelVas2008}, $\u$ is not given  but instead solves the Navier-Stokes equation. In principle, a system with Navier-Stokes equation is obviously more complicated; on the other hand, Navier-Stokes equation has an additional dissipation term in the energy relation. Also, the coupling term leads to a certain cancellation in the energy relation (the term $\iint \mathcal{U}\cdot (\mathcal{V}-\ve \mathcal{U})f_\ve$ in \eqref{energy_calculation_0} from Section~\ref{sec:zero-inertia}). 
The scaling (that is, how $\ve$ is introduced in Fokker-Planck equation) in \eqref{fp_with_u_and_diffusion_0_mr} is similar to \cite{GouJabVas2004}, but the boundary conditions in \cite{GouJabVas2004} are simpler (periodic).
In \cite{MelVas2008} the asymptotic regime of the Fokker-Planck equation is studied for the reflection boundary condition (spherical particles elastically collide with walls) as well, but due to no slip conditions for $\boldsymbol{u}$ (which imply no flow at boundary $\Gamma$), the limiting boundary conditions were not investigated.

In experimental observations, rod-like microswimmers, such as bacteria or bi-metallic particles, are more likely to spontaneously turn rather than jump to another position. 
In other words, random forces along the perimeter of a microswimmer caused by collisions with molecules of the fluid likely result in a significant torque whereas a net force is small.   
These observations imply that the translational diffusion coefficient is small, $D_{\text{tr}}\ll 1$. Hence, it is natural to study the limit $D_{\text{tr}}\to 0$ which is singular in the case of active particles, that is one cannot simply set $D_{\text{tr}}$ to zero in both Fokker-Planck equation \eqref{limiting_fp_mr} and no-flux condition \eqref{no_flux_bc_0_mr} in order to obtain the limiting system as $D_{\text{tr}}\to 0$. This is because active particles tend to accumulate at walls and in particular they form a boundary layer of the width $\sim D_{\text{tr}}$. Moreover, from boundary conditions \eqref{no_flux_bc_0_mr} it follows that $\nabla \rho$ may blow up near boundary $\Gamma$ in the limit $D_{\text{tr}}\to 0$; in this case  for stable numerical simulation of \eqref{limiting_fp_mr}-\eqref{no_flux_bc_0_mr} with small $D_{\text{tr}}$ one needs a very fine mesh resulting in high computational complexity of the simulations. 

 To investigate vanishing translational diffusion $D_{\text{tr}}$ we re-denote the translational diffusion coefficient by symbol $\delta$ which is typically used for notations of small parameters: 
\begin{equation*}
\delta:=D_{\mathrm{tr}}.
\end{equation*}

Then problem \eqref{limiting_fp_mr}-\eqref{no_flux_bc_0_mr} consists of the Fokker-Planck equation
\begin{equation}\label{FP_no_diff}
\partial_t \rho+ \nabla_{\boldsymbol{r}}\cdot (\boldsymbol{u}\rho) +\partial_{\varphi}(T\rho)=\delta\Delta_{\boldsymbol{ r}}\rho+D_{\text{rot}}\partial^2_{\varphi}\rho
\end{equation}
and the boundary condition
\begin{equation}\label{no_flux_bc_b_layer}
\delta\dfrac{\partial \rho}{\partial \boldsymbol{n}}=(\boldsymbol{u}\cdot \boldsymbol{n})\rho,\quad \boldsymbol{r}\text{ on }\Gamma,\quad-\pi\leq \varphi< \pi.
\end{equation}       

{\it Our second main result} is the derivation of the limit in problem \eqref{FP_no_diff}-\eqref{no_flux_bc_b_layer} as $\delta \to 0$. The derivation is done by formal multi-scale asymptotic expansion; we formulate the result as a conjecture since rigorous justification of the multi-scale asymptotic expansions are out of the scope of this work.     

\smallskip 

\noindent{\bf Conjecture:} {\it In the limit $\delta\to 0$ the probability distribution function $\rho(t, \boldsymbol{r},\varphi)$ has the following representation: 
\begin{equation}
\label{limiting_representation}
\rho(t,\boldsymbol{r},\varphi)=\psi_{\mathrm{wall}}(t,\boldsymbol{r},\varphi)\, \delta_{\Gamma}(\boldsymbol{ r})+\rho_{\mathrm{bulk}}(t,\boldsymbol{r},\varphi).
\end{equation}
Here $\delta_{\Gamma}(\boldsymbol{r})$ is the $\delta$-function distribution supported on $\Gamma=\left\{\boldsymbol{\gamma}(s):0\leq s \leq L\right\}$ ($s$ is arc-length parameter of curve $\Gamma$) and 
probability distribution functions $\rho_{\text{bulk}}$ and $\psi_{\text{wall}}$ solve the following system: 
	\begin{empheq}[right=\empheqrbrace]{align}
\nonumber&\partial_t \rho_{\mathrm{bulk}}+\nabla_{\boldsymbol{r}}\cdot (\boldsymbol{u}\rho_{\mathrm{bulk}})+\partial_{\varphi}(T\rho_{\mathrm{bulk}})=D_{\mathrm{rot}}\partial^2_{\varphi}\rho_{\mathrm{bulk}}+\sum\limits_{i=1,2}\chi_i\,\delta_{\Gamma}({\boldsymbol{r}})\delta(\varphi-\varphi_i), \\
\nonumber& \hspace{275pt}{\boldsymbol{r}}\in \Omega, \quad -\pi\leq \varphi<\pi,\\
\nonumber&\partial_{t}\psi_{\mathrm{wall}}+\partial_s((\boldsymbol{u}\cdot \boldsymbol{\tau}) \psi_{\mathrm{wall}})+\partial_\varphi \left(T \psi_{\mathrm{wall}}\right)= D_{\mathrm{rot}}\partial^2_\varphi \psi_{\mathrm{wall}}+(\boldsymbol{u}\cdot \boldsymbol{n})\rho_{\mathrm{bulk}},\\ &\hspace{245 pt}\boldsymbol{r}\in\Gamma, \, s\in [0,L],\, \varphi\in (\varphi_1,\varphi_2),\label{eqn:no-diff}\\
\nonumber&\rho_{\mathrm{bulk}}=0,\quad\boldsymbol{r}\in\Gamma,\quad \varphi\notin [\varphi_1,\varphi_2],\quad \rho_{\mathrm{bulk}} \,\mathrm{ is }\, 2\pi\, \mathrm{-periodic\,in}\,\varphi\,\mathrm{ for\, all }\, \boldsymbol{r}\in \Omega,\\
\nonumber&\psi_{\mathrm{wall}}=0,\quad s\in [0,L],\quad \varphi\in\left\{\varphi_1,\varphi_2\right\},\\
\nonumber&\chi_i(s,\varphi):=(-1)^i\left(T \psi_{\mathrm{wall}}-D_{\mathrm{rot}}\partial_{\varphi}\psi_{\mathrm{wall}}\right),\quad i=1,2.
\end{empheq}
Here $\boldsymbol{\tau}$ denotes the tangential vector of $\Gamma$ and angles $\varphi_1(s)$ and $\varphi_2(s)$ are introduced such that 
\begin{equation*}
\text{$\boldsymbol{u}\cdot \boldsymbol{n}>0$ for $\varphi\in (\varphi_1(s),\varphi_2(s))$, and $\boldsymbol{u}\cdot \boldsymbol{n}\leq 0$ otherwise.}
\end{equation*}  
}

Representation \eqref{limiting_representation} means that the total probability distribution function $\rho$ consists of the regular part, $\rho_{\text{bulk}}$, describing distribution of particles in the bulk, and the singular part, $\psi_{\text{wall}}\,\delta_{\Gamma}$, describing distribution of particles accumulated at wall. Derivation of system \eqref{eqn:no-diff} is presented in Section~\ref{sec:nodiffusion}. We also provide a numerical example in Section~\ref{sec:numerics} in which we test the derived kinetic approaches \eqref{FP_no_diff}-\eqref{no_flux_bc_b_layer} and \eqref{eqn:no-diff} with results of Monte-Carlo simulations for the corresponding individual based model.


\section{Vanishing inertia limit in Fokker-Planck equation}
\label{sec:zero-inertia}

\noindent {\it Proof of Theorem \ref{thm:main_mr}.} In this section we take $D_{\text{tr}}=D_{\text{rot}}=1$ for the sake of simplicity.  To consider the limit $\ve\to 0$, introduce the mean flux (or, the mean velocity), and ``the kinetic pressure": 
\begin{equation*}
J_\ve (t,\mathcal{X}):=\dfrac{1}{\ve}\int_{\mathbb R^3} {\mathcal{V}} f_\ve  \,\mathrm{d}\mathcal{V}\quad\text{and}\quad\mathbb P_\ve(t,\mathcal{X}):=\int_{\mathbb R^3}{\mathcal{V}}\otimes {\mathcal{V}} f_\ve\, \text{d}{\mathcal{V}}.
\end{equation*}

\noindent By integration of \eqref{fp_with_u_and_diffusion_0_mr} with respect to $\text{d}\mathcal{V}$  and $\varepsilon \mathcal{V}\,\text{d}\mathcal{V}$ one obtains the system for $\rho_\ve$ and $J_\ve$:
\begin{align}
&  \partial_t \rho_\ve +\nabla_{\mathcal{X}}\cdot J_\ve =0,\label{eqrho}\\
\ve^2& \partial_t J_\ve +\nabla_{\mathcal{X}} \cdot \mathbb P_\ve=\rho_\ve \,\mathcal{U} - J_\ve.  \label{eqJ}
\end{align}
By using arguments similar to \cite{GouJabVas2004}, we will show that the limit of \eqref{eqJ} is $J=\rho \,\mathcal{U} -\nabla_{\mathcal{X}} \rho$. Substitution of this formula for $J$ into the limiting version of the equation \eqref{eqrho} (that is, equation \eqref{eqrho} without sub-indexes $\ve$), then gives:
\begin{equation*}
\partial_t \rho +\nabla_{\mathcal{X}}\cdot (\mathcal{U}\rho) =\Delta_{\mathcal{X}}\rho.
\end{equation*} 
The main question is how to find the boundary condition for $\rho$. Note that from  collision boundary condition \eqref{bc_fokker-planck_mr} it follows that active rods cannot leave domain $\Omega$, that is, there is no flux through the boundary $\Gamma$: 
\begin{equation*}
\hat{J}_\ve\cdot\boldsymbol{n}=0 \text{ on }\Gamma, 
\end{equation*}
where $\hat{J}_\ve=(J_1,J_2)$ (no $J_3$, corresponding to the flux of orientations $\varphi$). 
The main purpose of this section is to prove that this relation is preserved in the limit $\ve\to 0$, which, taking into account the formula for the limiting flux $J=\rho \,\mathcal{U} -\nabla_{\mathcal{X}} \rho$, is equivalent to 
\begin{equation}\label{no_flux_bc}
\dfrac{\partial \rho}{\partial \boldsymbol{n}}=(\boldsymbol{u}\cdot \boldsymbol{n})\rho.
\end{equation}     
In order to prove Theorem~\ref{thm:main_mr} we will use two auxiliary propositions. In Proposition~\ref{prop:entropy}, the energy estimate is established. This estimate  leads to a priori bounds needed to obtain that the family $\left\{f_\ve \right\}_\ve$ has a limit as $\ve \to 0$ (Proposition \ref{prop2}).   

 First introduce the following notations: 
 \begin{eqnarray*}
 \mathcal{E}_\ve(t)&:=&\int_{\mathbb R^2\times \mathbb R}\int_{ \Omega\times(-\pi,\pi)}\left\{\frac{\boldsymbol{v}^2}{2}+\frac{\omega^2}{2}+\ln f_\ve\right\}f_\ve \, \text{d}\mathcal{X}\text{d}{\mathcal{V}},\\
 d_\ve(t,\mathcal{X},\mathcal{V})&:=&\left(({\mathcal{V}}-\ve \mathcal{U})+\nabla_{{\mathcal{V}}}(\ln f_\ve)\right) \sqrt{f_\ve}. 
 \end{eqnarray*}
 Recall that $\mathcal{X}=(\boldsymbol{r}, \varphi)\in \Omega\times [-\pi,\pi)$ and $\mathcal{V}=(\boldsymbol{v},\omega)\in\mathbb R^2_{\boldsymbol{v}} \times \mathbb R_\omega=\mathbb R^3$.
 
 \begin{proposition}\label{prop:entropy}
 	There exists a constant $C$, independent of $\ve$, such that the following estimate (the entropy inequality) holds:
 	\begin{equation}\label{reflect_entropy}
 	\dfrac{\mathrm{d}}{\mathrm{d}t}\mathcal{E}_{\ve}(t) +\dfrac{1}{2\ve^2}\int_{\mathbb R^3}\int_{\Omega\times (-\pi,\pi)}|d_\ve|^2 \,\mathrm{d}\mathcal{X}\mathrm{d}{\mathcal{V}} 
 	<C.
 	\end{equation}
 \end{proposition}

\begin{proof}
	Multiplication of \eqref{fp_with_u_and_diffusion_0_mr}  by 
	$$\dfrac{\boldsymbol{v}^2}{2}+\dfrac{\omega^2}{2}+\ln f_\ve $$ 
	and integration with respect to both $\mathcal{X}$ and ${\mathcal{V}}$ gives:
	\begin{eqnarray}
		\nonumber&&\dfrac{\text{d}}{\text{d}t}\left[\int_{\mathbb R^3}\int_{\Omega\times (-\pi,\pi)}\left\{\frac{\boldsymbol{v}^2}{2}+\frac{\omega^2}{2}+\ln f_\ve\right\}f_\ve\,\text{d}\mathcal{X}\text{d}\mathcal{V}\right]=\\ 
		\nonumber&&\hspace{50 pt}=
		-\dfrac{1}{\ve}\int_{\mathbb R^3}\int_{-\pi}^{\pi}\int_{\Gamma}\left\{\frac{\boldsymbol{v}^2}{2}+\frac{\omega^2}{2}+\ln f_\ve \right\}(\boldsymbol{v}\cdot \boldsymbol{n}) f_\ve \,\text{d}s_{\boldsymbol{r}}\,\text{d}\varphi\,\text{d}\mathcal{V} \\
		\nonumber&&\hspace{60 pt}-\dfrac{1}{\ve^2}\int_{\mathbb R^3}\int_{\Omega\times(-\pi,\pi)}
		\left|({\mathcal{V}}-\ve \mathcal{U})+\nabla_{\mathcal{V}}(\ln f_\ve)\right|^2 f_\ve \,\text{d}\mathcal{X}\text{d}\mathcal{V}\\
		&&\hspace{60pt} -\dfrac{1}{\ve}\int_{\mathbb R^3}\int_{\Omega\times(-\pi,\pi)}
		\mathcal{U} \cdot \left(\mathcal{V}-\ve \mathcal{U}\right) f_\ve \,\text{d}\mathcal{X}\text{d}\mathcal{V}.\label{energy_calculation_0}
	\end{eqnarray}

	\noindent Next we compute the boundary term in the right hand side of \eqref{energy_calculation_0} (the one with the integral over $\Gamma$). To this end, we use the boundary condition \eqref{bc_fokker-planck_mr}. For each $\boldsymbol{r}\in \Gamma$ and $-\pi\leq\varphi<\pi$ denote 
	\begin{equation}\label{def_of_S}
	S_{\sigma>0}:=\left\{(\boldsymbol{v},\omega): \sigma>0 \right\}\text{ and }
		S_{\sigma<0}:=\left\{(\boldsymbol{v},\omega): \sigma<0 \right\},
	\end{equation}
	where $\sigma=\sigma(\boldsymbol{r},\boldsymbol{v},\varphi,\omega)$ is given by \eqref{def_of_sigma} in Section \ref{sec:individual-rod}. The introduced sets $S_{\sigma>0}$ and $S_{\sigma<0}$ can be understood as sets of configurations (velocities) of an active rod before and after a collision, respectively, at the given location of the boundary $\boldsymbol{r}\in \Gamma$ and the given orientation $\varphi\in [-\pi,\pi)$. Then the boundary integral can be written as follows for each $\boldsymbol{r}\in \Gamma$ and $-\pi\leq \varphi<\pi$:   
	\begin{eqnarray}
		\nonumber&&\int_{\mathbb R^3}\left\{\frac{\boldsymbol{v}^2}{2}+\frac{\omega^2}{2}+\ln f_\ve \right\}(\boldsymbol{v}\cdot \boldsymbol{n}) f_\ve \,\text{d}\mathcal{V} =\\
		\nonumber&&\hspace{60pt} \int\limits_{S_{\sigma< 0}}\left\{\frac{\boldsymbol{v}^2}{2}+\frac{\omega^2}{2}+\ln f_\ve \right\}(\boldsymbol{v}\cdot \boldsymbol{n}) f_\ve \,\text{d}\mathcal{V}
		\\\nonumber&&\hspace{60pt}
		+\int\limits_{S_{\sigma>0}}\left\{\frac{\boldsymbol{v}^2}{2}+\frac{\omega^2}{2}+\ln f_\ve \right\}(\boldsymbol{v}\cdot \boldsymbol{n}) f_\ve \,\text{d}\mathcal{V}.
	\end{eqnarray}
	We claim that the two integrals in the right hand side of the equality above cancel each other. To verify this, one needs to make the substitution in the first integral $\mathcal{V}'=\mathcal{C}\mathcal{V}$ with $\mathcal{C}$ from \eqref{def_of_C} (or equivalently \eqref{collision_rule_velocities}-\eqref{collision_rule_angular_velocities}) and to use boundary condition \eqref{bc_fokker-planck_mr}, conservation of energy during a collision \eqref{conservation_of_kinetic_energy}, and $\mathrm{d}\mathcal{V}'=\mathrm{d}\mathcal{V}$ which follows from $|\text{det}\,\mathcal{C}|=1$ (see \eqref{collision_relation}). Hence, the boundary term in the right hand side of \eqref{energy_calculation_0} vanishes.
	
	Note that in the same manner one can show that the conservation of total $f_\ve$: 
	\begin{equation*}
	\dfrac{\text{d}}{\text{d}t}\left[\int_{\mathbb R^3}\int_{\Omega\times (-\pi,\pi)}f_\ve\,\text{d}\mathcal{X}\text{d}\mathcal{V}\right]=0.
	\end{equation*}
	Indeed, by using equation \eqref{fp_with_u_and_diffusion_0_mr} and integration by parts it follows that 
	\begin{equation*}
	\dfrac{\text{d}}{\text{d}t}\left[\int_{\mathbb R^3}\int_{\Omega\times (-\pi,\pi)}f_\ve\,\text{d}\mathcal{X}\text{d}\mathcal{V}\right]=		-\dfrac{1}{\ve}\int_{\mathbb R^3}\int_{-\pi}^{\pi}\int_{\Gamma}(\boldsymbol{v}\cdot \boldsymbol{n}) f_\ve \,\text{d}s_{\boldsymbol{r}}\,\text{d}\varphi\,\text{d}\mathcal{V},
	\end{equation*}
	and one can show that the right hand side vanishes following the same arguments as for the boundary term in \eqref{energy_calculation_0}
	
	Finally, the last term in the right hand side of \eqref{energy_calculation_0} is estimated as follows: 
	\begin{eqnarray}
	&&-\dfrac{1}{\ve}\int_{\mathbb R^3}\int_{\Omega\times(-\pi,\pi)}
	\mathcal{U} \cdot \left(\mathcal{V}-\ve \mathcal{U}\right) f_\ve \,\text{d}\mathcal{X}\text{d}\mathcal{V}\nonumber\\
	&&\hspace{60pt}=-\dfrac{1}{\ve}\int_{\mathbb R^3}\int_{\Omega\times(-\pi,\pi)}
	\mathcal{U} \cdot \left((\mathcal{V}-\ve \mathcal{U})+\nabla_{\mathcal{V}}(\ln f_\ve)\right) f_\ve \,\text{d}\mathcal{X}\text{d}\mathcal{V}\nonumber\\
	&&\hspace{60pt}\leq C +\dfrac{1}{2\ve^2}\int_{\mathbb R^3}\int_{\Omega\times (-\pi,\pi)}|d_\ve|^2 \,\mathrm{d}\mathcal{X}\mathrm{d}{\mathcal{V}}. \label{prop2_last_ineq}
	\end{eqnarray} 
	Thus, we obtained \eqref{reflect_entropy} and the proposition is proved. 
\end{proof}

\bigskip 

\noindent In the standard manner (see, e.g., \cite{GouJabVas2004}) the entropy estimate \eqref{reflect_entropy}  implies the following bounds: 
\begin{eqnarray*}
f_\ve \left(1+\frac{ \boldsymbol{v}^2}{2}+\frac{\omega^2}{2}+|\ln f_\ve|\right) &\text{ is bounded in }&L^{\infty} (0,T;L^1(\Omega\times(-\pi,\pi)\times \mathbb R^3)),\\
\ve^{-1}d_\ve &\text{ is bounded in }&L^{2} (0,T;L^2(\Omega\times(-\pi,\pi)\times \mathbb R^3)),\\
\rho_\ve  &\text{ is bounded in }&L^{\infty} (0,T;L^1(\Omega\times(-\pi,\pi))),\\
J_\ve \text{ and }J_\ve-\rho_\ve \mathcal{U} &\text{ are bounded in }&L^{2} (0,T;L^1(\Omega\times(-\pi,\pi))).
\end{eqnarray*}
The proof of the following proposition is also standard and can be found in \cite{GouJabVas2004}.
\begin{proposition}\label{prop2}
	There exist such $\rho$ and $J$ that the following convergences hold as $\ve\to 0$ in the distributional sense: 
	\begin{eqnarray}
	&& \rho_\ve \rightharpoonup \rho,\,\,\, J_\ve \rightharpoonup J,\,\,\,\mathbb P_\ve \rightharpoonup \rho \mathbb I.
	\end{eqnarray}
	Here $\mathbb I$ is identity matrix.
\end{proposition}

\noindent Now we are in position to complete the proof of Theorem \ref{thm:main_mr}.

\begin{proof}[Back to proof of Theorem \ref{thm:main_mr}]
By multiplication of \eqref{fp_with_u_and_diffusion_0_mr} 
by a scalar test functions $\psi(t,\mathcal{X},{\mathcal{V}})$ with a finite support in $0<t<T$, $\mathcal{X}\in \Omega\times (-\pi,\pi)$ and ${\mathcal{V}}\in\mathbb R^3$, integration with respect to $t$, $\mathcal{X}$ and $\mathcal{V}$, as well as integration by parts, one obtains the following equality: 
\begin{eqnarray}
	&&\int_0^T\int_{\mathbb R^3}\int_{\Omega\times(-\pi,\pi)} f_\ve \, \left\{\partial_t\psi + \dfrac{1}{\ve}{\mathcal{V}}\cdot \nabla_{\mathcal{X}} \psi +\dfrac{1}{\ve^2}(\varepsilon \mathcal{U}-\mathcal{V})\cdot \nabla_{\mathcal{V}} \psi + \dfrac{1}{\ve^2}\Delta_{\mathcal{V}} \psi \right\}\,\text{d}\mathcal{X}\text{d}\mathcal{V}\text{d}t \nonumber\\
	&& \hspace{120 pt}-\dfrac{1}{\ve}\int_{\Gamma}\int_{-\pi}^{\pi} \int_{S_{\sigma >0}}(\boldsymbol{v}\cdot \boldsymbol{n})f_\ve (\psi-\psi')\, \text{d}{\mathcal{V}} \text{d}\varphi\text{d}s_{\boldsymbol{r}}\text{d}t=0,\label{weak_solution_definition}
\end{eqnarray} 
where $\psi'=\psi(t,\mathcal{X},\mathcal{V}')$  with $\mathcal{V}'=\mathcal{C}\mathcal{V}$, and matrix $\mathcal{C}$ and set $S_{\sigma>0}$ are defined by \eqref{def_of_C} and \eqref{def_of_S}, respectively. Equality \eqref{weak_solution_definition} can be understood as the weak formulation of the problem \eqref{fp_with_u_and_diffusion_0_mr} with boundary condition \eqref{bc_fokker-planck_mr}. 

Next, take a test function in \eqref{weak_solution_definition} which is independent of $\mathcal{V}$: $\psi:=g(t,\mathcal{X})$ (for the sake of clarity, choose different symbol for the test function here: $g$ instead of $\psi$). Formally, such a test function is not admissible since it does not have a finite support in $\mathcal{V}$. On the other hand, one can use truncations $g(t,\mathcal{X})\chi_{|\mathcal{V}|<R}(\mathcal{V})$ as test functions and pass to the limit $R\to \infty$ to obtain \eqref{weak_solution_definition} for $\psi=g(t,\mathcal{X})$.
By integrating in ${\mathcal{V}}$ we obtain: 
\begin{equation*}
\int_{0}^T\int_{\Omega\times(-\pi,\pi)} \rho_\ve  \partial_t g + J_\ve \cdot \nabla_{\mathcal{X}} g \,\text{d}\mathcal{X}\text{d}t=0.
\end{equation*}
Note that the boundary term in \eqref{weak_solution_definition} vanishes for test functions independent from $\mathcal{V}$. By passing to the limit $\ve \to 0$ we obtain: 
\begin{equation}\label{eq_for_weak_rho}
\int_{0}^T\int_{\Omega\times(-\pi,\pi)} \rho  \partial_t g + J \cdot \nabla_{\mathcal{X}} g \,\text{d}\mathcal{X}\text{d}t=0.
\end{equation}

\noindent Take the test function of the form $\psi:=\ve {v}_i h_i(t,\mathcal{X})$ for $i=1,2$:
\begin{eqnarray}\label{prelim_weak_def_for_J}
&&\int_{0}^T\int_{\Omega\times(-\pi,\pi)} \ve^2 J_\ve^i \partial_t h_i + \mathbb P_\ve^{ij} \partial_{\mathcal{X}_j}h_i + (\rho_\ve \mathcal{U}^i - J_\ve^i)h_i \,\text{d}\mathcal{X}\text{d}t \nonumber\\
&&\hspace{50pt} -2\int_{0}^{T}\int_{\mathbb{R}}\int_{-\pi}^{\pi}\int_{\Gamma}\int_{S_{\sigma>0}} \sigma(\boldsymbol{v}\cdot \boldsymbol{n})f_\ve h_i n_i \,\text{d}v\text{d}s_x\text{d}\varphi\text{d}\omega \text{d}t=0.  
\end{eqnarray}
Here both the super- and sub-index $i$ stand for the coordinate number. Consider $\boldsymbol{h}=\left\{h_i\right\}_{i=1}^{3}$ so that $\boldsymbol{h}\cdot\boldsymbol{n}=h_1n_1+h_2n_2=0$ for all $-\pi\leq \varphi < \pi$ and $h_3(t,\mathcal{X})$ is arbitrary. The test function $\psi$ corresponding to $h_3(t,\mathcal{X})$ is $\psi=\omega h_3(t,\mathcal{X})$, and in this case equality \eqref{prelim_weak_def_for_J} holds for $i=3$ with no second (boundary) term. Then taking the sum with respect to $i$ in \eqref{prelim_weak_def_for_J} leads to that the boundary term in \eqref{prelim_weak_def_for_J} vanishes, so that the following equality holds: 
\begin{equation*}
\int_{0}^T\int_{\Omega\times(-\pi,\pi)} \ve^2 J_\ve \cdot \partial_t\boldsymbol{h} + \mathbb P_\ve : \nabla_{\mathcal{X}} \boldsymbol{h}+(\rho_\ve \mathcal{U} - J_\ve)\cdot \boldsymbol{h} \,\text{d}\mathcal{X}\text{d}t=0. 
\end{equation*} 
Passing to the limit $\ve \to 0$ and using Proposition \ref{prop2} we get:
\begin{equation}\label{eq_for_J}
\int_{0}^T\int_{\Omega\times(-\pi,\pi)} \rho\mathbb I : \nabla_{\mathcal{X}} \boldsymbol{h}+(\rho \mathcal{U} - J)\cdot \boldsymbol{h} \,\text{d}\mathcal{X}\text{d}t=0. 
\end{equation}  
This equality gives us the following relation for $J$: 
\begin{equation}
\label{weak_formula_for_J}
\int_{0}^T\int_{\Omega\times(-\pi,\pi)} J\cdot \boldsymbol h \,\text{d}\mathcal{X}\text{d}t= 
\int_{0}^T\int_{\Omega\times(-\pi,\pi)} \rho\mathbb I : \nabla_{\mathcal{X}} \boldsymbol{h}+\rho \mathcal{U}\cdot \boldsymbol{h} \,\text{d}\mathcal{X}\text{d}t
\end{equation}
 for all $\boldsymbol{h}$ such that $h_1n_1+h_2n_2|_{\Gamma}=0$ for all $-\pi\leq \varphi < \pi$.
Finally, take any $g(t,\mathcal{X})$ in \eqref{eq_for_weak_rho} such that $\left.\dfrac{\partial g}{\partial \boldsymbol{n}}\right|_{\Gamma}=0$ for all $-\pi\leq \varphi < \pi$ and $\boldsymbol{h}:=\nabla_{\mathcal{X}} g$ in \eqref{weak_formula_for_J} to express $\int_0^T\int_{\Omega\times(-\pi,\pi)} J\cdot \nabla_{\mathcal{X}} g \,\text{d}\mathcal{X}\text{d}t$ in \eqref{eq_for_weak_rho}. We obtain the following weak formulation for the equation for $\rho$: 
\begin{equation}\label{eq_for_rho}
\int_{0}^T\int_{\Omega\times(-\pi,\pi)} \left\{\rho  \partial_t g +\rho \mathcal{U}\cdot \nabla_{\mathcal{X}}g+ \rho\mathbb I : \nabla_{\mathcal{X}}^2 g \right\}\,\text{d}\mathcal{X}\text{d}t=0 
\end{equation}
for all $g$  such that $\left.\dfrac{\partial g}{\partial\boldsymbol{n}}\right|_{\Gamma}=0$ for all $-\pi\leq \varphi < \pi$.
By integrating by parts and using condition $\left.\dfrac{\partial g}{\partial \boldsymbol{n}}\right|_{\Gamma}=0$, \eqref{eq_for_rho} leads to 
\begin{eqnarray}\nonumber 
&&\int_{0}^T\int_{\Omega\times(-\pi,\pi)} g\left\{ \partial_t \rho +\nabla_{\mathcal{X}}\cdot (\rho \mathcal{U})-\Delta_{\mathcal{X}} \rho\right\} \,\text{d}\mathcal{X}\text{d}t+
\\&&\hspace{60 pt} +\int_{0}^{T}\int_{-\pi}^{\pi}\int_{\Gamma}g\left\{-\dfrac{\partial \rho}{\partial \boldsymbol{n}}+(\boldsymbol{u}\cdot \boldsymbol{n})\rho\right\}\,\text{d}s_{\boldsymbol{ r}}\text{d}\varphi\text{d}t =0. \label{eq_for_rho_strong}
\end{eqnarray}
Varying $g$ on $\Gamma$ we get our final result which is the no-flux boundary condition \eqref{no_flux_bc}:
\begin{equation*}
\dfrac{\partial \rho}{\partial \boldsymbol{n}}=(\boldsymbol{u}\cdot \boldsymbol{n})\rho
\end{equation*} 
 for all $-\pi\leq \varphi < \pi$.
 
Thus, the proof of Theorem \ref{thm:main_mr} is complete.
\end{proof}

\section{Boundary layer equation at wall in vanishing translational diffusion limit: derivation of \eqref{eqn:no-diff}}
\label{sec:nodiffusion}

To describe the behavior of $\rho$ near the wall, one needs to consider the limiting behavior of $\rho$ inside the ``band" $\Omega^*:=\left\{\boldsymbol{r}\in\Omega: \text{dist}(\boldsymbol{r}, \Gamma)<c\right\}$ where $c$ is small but independent of $\delta$. Here we choose $c>0$ such that $c<\kappa_{\text{max}}^{-1}$ where $\kappa_{\text{max}}$ is the maximum curvature along $\Gamma$. If $\Gamma$ is a straight line or a segment, then curvature $\kappa$ is zero, and $c$ is an arbitrary number independent of $\delta$.   

We  introduce new coordinate system in band $\Omega^{*}$, related to parametrization of wall $\Gamma=\partial \Omega$. Namely, let $\Gamma=\{\boldsymbol{\gamma}(s)\colon 0\leq s\leq L\}$ be the natural parametrization of the wall $\Gamma$ (in other words, $s$ is the arc length parameter) and $\boldsymbol{n}(s)$, $\boldsymbol{\tau}(s)$ be an outward normal and tangential vectors at $\boldsymbol{\gamma}(s)\in \Gamma$, respectively. For every $\boldsymbol{r}\in\Omega^{*}$, we define 
\begin{equation}\label{change_of_variables}
\boldsymbol{r}=\boldsymbol{r}(r,s)=\boldsymbol{\gamma}(s)-r\boldsymbol{n}(s),
\end{equation}
 where $r=\text{dist}(\boldsymbol{r},\Gamma)$ and $\boldsymbol{\gamma}(s)$  is the ``projection" of $\boldsymbol{r}$ onto $\Gamma$. 
 
 In this new coordinate system we introduce the two-scale ansatz for unknown function $\rho$:
 \begin{equation}\label{twoscale}
 \rho=\rho_{w}+\rho_{b}= \sum\limits_{k=-1}^{\infty} \delta^{k} \rho_{w}^{(k)}(t,\delta^{-1}r,s,\varphi)+\sum\limits_{k=0}^{\infty}\delta^{k}\rho_b^{(k)}(t,r,s,\varphi).
 \end{equation}
Here sub-indexes '$w$' and '$b$' stand for ``wall" and ``bulk", respectively. Variable $z$ denotes below the second argument of $\rho_{w}^{(k)}$, the $k$th wall (boundary layer) coefficient, {\it i.e.,} $z=\delta^{-1}r$ in \eqref{twoscale}. We assume that functions $\{\rho_{w}^{(k)}(t,z,s,\varphi)\}_{k}$ vanish with all derivatives in $z$ as $z\to \infty$.  
In what follows, we focus on the first three terms of two-scale expansion \eqref{twoscale} or, in other words, on terms of order $\delta^{-1}$ and $\delta^{0}$:
\begin{equation}\label{twoscale0}
\rho= \delta^{-1} \rho_{w}^{(-1)}(t,\delta^{-1}r,s,\varphi)+ \rho_{w}^{(0)}(t,\delta^{-1}r,s,\varphi)+\rho_b^{(0)}(t,r,s,\varphi)+O(\delta).
\end{equation}
The representations \eqref{twoscale0} and \eqref{limiting_representation} are related via the following equalities: 
\begin{equation}\label{def_of_rho_and_psi}
\psi_{\text{wall}}(t,\boldsymbol{ r},\varphi)=\int_0^{+\infty}\rho_{w}^{(-1)}(t,z,s,\varphi)\, \text{d}z\,\text{ and }\,\rho_{\text{bulk}}(t,\boldsymbol{ r},\varphi)=\rho_{b}^{(0)}(t,r,s,\varphi).
\end{equation}

Next, we rewrite the Fokker-Planck equation \eqref{FP_no_diff} in the coordinate system $(r,s)$. To this end, we introduce the inverse substitution functions $R(\boldsymbol{r})$ and $S(\boldsymbol{r})$: 
\begin{equation*}
\begin{array}{c}r=R(\boldsymbol{r})\\
s=S(\boldsymbol{r})
\end{array}
~~\Leftrightarrow~~
\begin{array}{c}
R(\boldsymbol{\gamma}(s)-r\boldsymbol{n}(s))=r,\\
S(\boldsymbol{\gamma}(s)-r\boldsymbol{n}(s))=s.
\end{array}
\end{equation*} 
Using chain rule and the 2D Frenet-Serret relation for the normal vector $\boldsymbol{n}'(s)=-\kappa(s) \boldsymbol{\tau}(s)$ where $\kappa(s)$ is the curvature of $\Gamma$ at $\boldsymbol{r}=\boldsymbol{\gamma}(s)$, one obtains
\begin{equation}\label{jacobian}
\nabla_{\boldsymbol{r}}R=-\boldsymbol{n}~~\text{ and }~~\nabla_{\boldsymbol{r}}S=(1-\kappa r)^{-1}\boldsymbol{\tau}.
\end{equation} 
To compute $\nabla_{\boldsymbol{r}}\rho$, $\Delta_{\boldsymbol{r}}\rho$ and $\nabla_{\boldsymbol{r}}\cdot \boldsymbol{u}$ one can use \eqref{jacobian} and both 2D Frenet-Serret relations, $\boldsymbol{n}'(s)=-\kappa(s) \boldsymbol{\tau}(s)$ and $\boldsymbol{\tau}'(s)=\kappa(s) \boldsymbol{n}(s)$:
\begin{eqnarray*}
\nabla_{\boldsymbol{r}} \rho&=& -\partial_r\rho \, \boldsymbol{n}+\frac{\partial_s \rho}{1-\kappa r} \,\boldsymbol{\tau}, \\
\Delta_{\boldsymbol{r}} \rho&=&\partial^2_r \rho- \frac{\kappa \partial_r \rho}{1-\kappa r}+\frac{r(\partial_s \kappa)(\partial_s \rho)}{(1-\kappa r)^3}+\frac{\partial^2_s \rho}{(1-\kappa r)^2}.\\
\nabla_{\boldsymbol{r}}\cdot u &=& -\partial_r u_{\text{n}}+\dfrac{\partial_s u_{\tau}}{1-\kappa r}+\dfrac{\kappa u_{\text{n}}}{1-\kappa r}.
\end{eqnarray*}
Here $u_{\text{n}}=\boldsymbol{u}\cdot \boldsymbol{n}$ and $u_{\tau}=\boldsymbol{u}\cdot \boldsymbol{\tau}$.

Then the original problem (\ref{FP_no_diff})-(\ref{no_flux_bc_b_layer}) converts into
\begin{eqnarray}\nonumber
	&&\partial_t \rho- u_{\text n}\partial_r \rho+\frac{u_{\tau}\partial_s\rho}{1-\kappa r} +\left(-\partial_r u_{\text{n}}+\dfrac{\partial_s u_{\tau}}{1-\kappa r}+\dfrac{\kappa u_{\text{n}}}{1-\kappa r}\right)\rho+\partial_{\varphi}(T\rho)\\
	&&\hspace{80pt}=\delta\left(\partial^2_r \rho- \frac{\kappa \partial_r \rho}{1-\kappa r}+\frac{r(\partial_s \kappa)(\partial_s\rho)}{(1-\kappa r)^3}+\frac{\partial^2_s \rho}{(1-\kappa r)^2}   \right)+D_{\text{rot}}\partial^2_{\varphi}\rho\label{FP_no_diff_b_layer}
\end{eqnarray}
with boundary conditions
\begin{equation}\label{no_flux_bc_b_layer_2}
	\delta\dfrac{\partial \rho}{\partial r}=-u_{\text{n}}\rho \,\text{ if } r=0.
\end{equation}

When substituting representation \eqref{twoscale0} into equation \eqref{FP_no_diff_b_layer}, we will treat the second term in the left hand side as follows: 
\begin{equation*}
u_{\text{n}}\partial_r\rho=\left(u_{\text{n}}^{(0)}+(z\delta) \, u_{\text{n}}^{(1)} +\frac{1}{2}(z\delta )^2u_{\text{n}}^{(2)}+...\right)\partial_r\rho_{w}+u_{\text{n}}(r,s)\,\partial_r \rho_{b},
\end{equation*}    
where $u_{\text{n}}^{(k)}=1/k!\, \partial^{k}_{r}u_{\text{n}}|_{r=0}$. Note that $u_{\text{n}}^{(k)}$ is a function of $\varphi$ and $s$ for each $k=1,2,...$. 

At the order $\delta^{-2}$ in \eqref{FP_no_diff_b_layer} one has the following equality: 
\begin{equation}\label{delta_minus2_pde}
u_{\text{n}}^{(0)}\partial_z\rho_{w}^{(-1)}+\partial^2_{z} \rho_{w}^{(-1)}=0.
\end{equation}
The lowest order in the boundary conditions \eqref{no_flux_bc_b_layer} is $\delta^{-1}$ and the corresponding equality is
\begin{equation}\label{delta_minus1_bc}
u_{\text{n}}^{(0)}\rho_{w}^{(-1)}+\partial_{z} \rho_{w}^{(-1)}=0, ~~z=0.
\end{equation}
Combining \eqref{delta_minus2_pde} and \eqref{delta_minus1_bc} we obtain a formula for $\rho_{w}^{(-1)}$:
\begin{equation*}
\rho_{w}^{(-1)}(t,z,s,\varphi)=\left\{\begin{array}{rl}B(t,s,\varphi)e^{-u_{\text{n}}^{(0)}z},&u_{\text{n}}^{(0)}<0,\\0, & u_{\text{n}}^{(0)}\geq 0.\end{array}\right.
\end{equation*}
At order $\delta^{-1}$, the Fokker-Planck equation \eqref{FP_no_diff_b_layer} has the form 
\begin{eqnarray}
u_{\text{n}}^{(0)}\partial_z\rho_{w}^{(0)}+\partial^2_{z} \rho_{w}^{(0)}& =&
\partial_t \left(Be^{-u_{\text{n}}^{(0)}z}\right) -u_{\text{n}}^{(0)}Be^{-u_{\text{n}}^{(0)}z}\left(-z u_{\text{n}}^{(1)}+\kappa \right)+u_{\tau}^{(0)}\partial_s\left( Be^{-u_{\text{n}}^{(0)}z} \right)\nonumber  \\
&&-\left(u_{\text{n}}^{(1)}-\partial_s u_{\tau}^{(0)}-\kappa u_{\text{n}}^{(0)}\right)Be^{-u_{\text{n}}^{(0)}z}\nonumber\\&&+\partial_{\varphi} \left( TBe^{-u_{\text{n}}^{(0)}z} \right) -D_{\text{rot}}\partial^2_{\varphi}  \left( Be^{-u_{\text{n}}^{(0)}z} \right). \label{eq_for_Beuz}
\end{eqnarray}
Boundary conditions at the order $\delta^{0}$ looks as follows: 
\begin{equation*}
u_{\text{n}}^{(0)} \rho_{w}^{(0)}+\partial_{z} \rho_{w}^{(0)}=-u_{\text{n}}^{(0)}\rho_{b}^{(0)}   \,\text{ for } z=0.
\end{equation*}
Consider $u_{\text{n}}^{(0)}>0$. After integration \eqref{eq_for_Beuz} with respect to $z$ from $0$ to $\infty$ and simplifications one obtains an equation for $\psi_{\text{wall}}=\int_{0}^{\infty}\rho_{w}^{(-1)}\,\text{d}z=B/u_{\text{n}}^{(0)}$:
\begin{equation}
u_{\text{n}}^{(0)}\rho_{b}^{(0)}=\partial_{t}\psi_{\text{wall}}+\partial_s(u_{\tau}^{(0)} \psi_{\text{wall}})+\partial_\varphi\left(T \psi_{\text{wall}}\right)- D_{\text{rot}}\partial^2_\varphi \psi_{\text{wall}}.\label{main_eq_no_diff}
\end{equation}
Equation \eqref{main_eq_no_diff} is a conservation law for the distribution of active rods $\psi_{\text{wall}}$ accumulated at wall; these active rods re-orient and move along the wall in the tangential direction. Term $u_{\text{n}}^{(0)}\rho_{b}^{(0)}$ in the left hind side of \eqref{main_eq_no_diff} accounts for particles coming from bulk. If divergence-free condition is imposed, $\nabla_{\bf r}\cdot u =0$, which is for $r=0$ has the form $u_{\text{n}}^{(1)}-\partial_s u_{\tau}^{(0)}-\kappa u_{\text{n}}^{(0)}=0$, then \eqref{main_eq_no_diff} can be rewritten in the form:       
\begin{equation}
u_{\text{n}}^{(0)}\rho_{b}^{(0)}=\partial_{t}\psi_{\text{wall}}+(u_{\text{n}}^{(1)}-\kappa u_{\text{n}}^{(0)})\psi_{\text{wall}}+u_{\tau}^{(0)}\partial_s \psi_{\text{wall}}+\partial_\varphi
\left( T \psi_{\text{wall}}\right)- D_{\text{rot}}\partial^2_\varphi \psi_{\text{wall}}.\nonumber
\end{equation}

If $u_{\text{n}}^{(0)}\leq 0$, then $\psi_{\text{wall}}=0$ and $\rho_b^{(0)}|_{r=0}=0$.

Next we obtain equation for $\rho_{\text{bulk}}=\rho^{(0)}_b$. To this end, consider equation \eqref{FP_no_diff_b_layer} at the order $\delta^{0}$ and pass to the limit $z\to \infty$. After we rewrite the resulting equation in the original coordinate system we recover the Fokker-Planck equation for $\rho_{\text{bulk}}$  
\begin{equation}\label{rho_bulk}
\partial_t \rho_{\text{bulk}}+\nabla_{\boldsymbol{r}}\cdot (\boldsymbol{u}\rho_{\text{bulk}})+\partial_{\varphi}(T\rho_{\text{bulk}})=D_{\text{rot}}\partial^2_{\varphi}\rho_{\text{bulk}}.
\end{equation}

Representation \eqref{limiting_representation} with $\psi_{\text{wall}}$ and $\rho_{\text{bulk}}$ satisfying equations \eqref{main_eq_no_diff} and \eqref{rho_bulk} respectively, is valid if $\boldsymbol{u}\cdot \boldsymbol{n}|_{\Gamma}>0$ for all $-\pi \leq \varphi <  \pi$. In this case, active rods accumulate at the wall, but they cannot leave the wall. On the other hand, due to the rotational diffusion, active rods in experiments (for example those described by system    
\eqref{force_balance}-\eqref{torque_balance}) may leave the wall and be re-injected into the bulk. In this case, one needs an additional boundary layer term in the representation \eqref{twoscale} (with an additional scale different from $1$ and $\delta$) which will capture active rods at wall $\Gamma$ with $\boldsymbol{u}\cdot \boldsymbol{n}\approx 0$. It is similar to ``parabolic boundary layers" in vanishing diffusion limit in elliptic equations; these boundary layers introduce two new scales $\delta^{1/3}$ and $\delta^{2/3}$, and these terms are constructed at boundary points where characteristics of the limiting hyperbolic equation are tangential to the boundary, see Sec.~2.7.5 in \cite{Hol2012}, see also \cite{van1976}. In this work, the formula for $\chi$ in \eqref{eqn:no-diff}, the flux of re-injected active rods, {\it i.e.,}, with $\boldsymbol{r}\in\Gamma$, $\boldsymbol{u}\cdot \boldsymbol{n}|_{\Gamma}\approx 0$ and $\dfrac{\text{d}}{\text{d}t}\boldsymbol{u}\cdot \boldsymbol{n}|_{\Gamma}<0$, is derived from the conservation of total density:
\begin{equation*}
\dfrac{\text{d}}{\text{d}t}\left[\int\limits_{-\pi}^{\pi}\int\limits_{\Omega}\,\rho_{\text{bulk}}\,\text{d}{\boldsymbol{r}}\text{d}\varphi+\int\limits_{\Gamma}\int\limits_{\varphi_1}^{\varphi_2}\psi_{\text{wall}}\,\text{d}\varphi\text{d}s\right]=0.
\end{equation*}  
Adding terms with flux $\chi$ to the right hand side of equation \eqref{rho_bulk} we derive system \eqref{eqn:no-diff}.

\bigskip

\section{Numerical example}
\label{sec:numerics}

\begin{figure}[t]
	\begin{center}
		\includegraphics[width=0.4\textwidth]{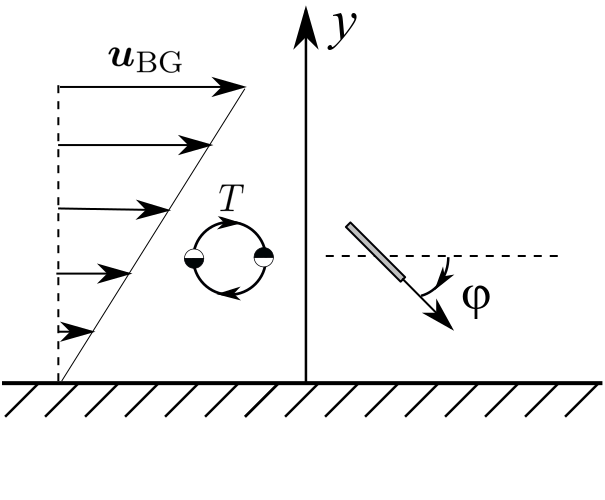}
		\includegraphics[width=0.5\textwidth]{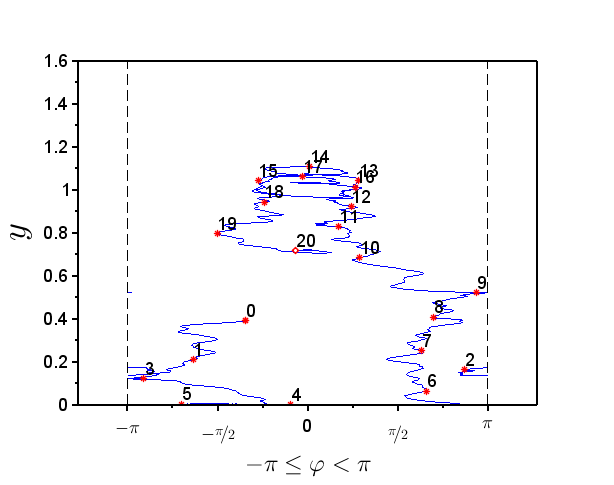}
		\caption{Left: sketch of an active rod with orientation $\varphi$ in $xy$-plane; $x$-axis represents the wall; straight arrows illustrate background shear flow $\boldsymbol{u}_{\text{BG}}$; circle shows how the background torque $T(\varphi)$ acts on the active rod; arrows along the circle show that equation $\dot{\varphi}=T(\varphi)$ has two semi-stable states $\varphi=0$ and $\varphi=-\pi$.
			Right: a sample trajectory in $\varphi y$-plane for $0\leq t \leq 20$. Red dots with numbers along trajectories indicate trajectory points at corresponding integer time moments. 
			\label{fig:sketch}
		}
	\end{center}
\end{figure}	

In this section, we provide a numerical example to illustrate the relation between a specific individual based model for an active rod and corresponding kinetic approaches discussed above. 

We assume that the wall $\Gamma$ coincides with $x$-axis, {\it i.e.} $\Gamma=\left\{(x,y):y=0\right\}$, and the active rod's probability distribution function does not depend on $x$ (see Fig.~\ref{fig:sketch}, left). Drag $\boldsymbol{u}$ exerted on an active rod is the sum of two components: a background shear flow $\boldsymbol{u}_{\text{BG}}=(\dot\gamma y, 0)$ and self-propulsion $v_{\text{prop}}(\cos(\varphi),\sin\varphi)$. The vertical component $u$ of drag velocity $\boldsymbol{u}$ and torque $T$ exerted on an active rod,  are defined as follows: 
\begin{equation*}
u(y,\varphi)=u(\varphi)= v_{\text{prop}}\sin \varphi, \quad T(y,\varphi)=T(\varphi)=-0.5\dot\gamma(1-\cos(2\varphi)).
\end{equation*} 
Note that since the background shear flow $\boldsymbol{u}_{\text{BG}}$ has zero $y$-component it does not enter the formula for vertical drag $u$. Expression for $T$ follows from \eqref{expression_for_torque} from Section~\ref{sec:individual-rod} with $T=\Phi_{\text{BG}}$. Parameters $v_{\text{prop}}=0.2$ and $\dot{\gamma}=1.0$ are self-propulsion speed and shear rate, respectively.  

First, we perform Monte Carlo simulations for the individual based model for an active rod with vertical component $y(t)$ of location $\boldsymbol{r}(t)$ and orientation angle $\varphi$ swimming in $\Omega~=~\left\{y>0\right\}$:
\begin{equation}\label{ibm1}
\dot{y}=u(\varphi),\quad \dot{\varphi}=T(\varphi)+\sqrt{2D_{\text{rot}}}\zeta,
\end{equation} 
where $\zeta$ is the white noise with $\langle \zeta(t),\zeta(t')\rangle=\delta(t-t')$ and $D_{\text{rot}}$ is the rotational diffusion coefficient.  The following ``overdamped" collision-with-wall rule for $y(t)=0$ is imposed:  
\begin{equation}\label{ibm2}
\dot{y}|_{y(t)=0}=\left\{\begin{array}{ll}0,& \sin \varphi\leq 0,\\ u(\varphi),& \sin \varphi >0, \end{array}\right.\quad 
\dot{\varphi}|_{y(t)=0}=T(\varphi)+\sqrt{2D_{\text{rot}}}\zeta.
\end{equation}
This collision rule means that the active rod does not move if it is oriented towards the wall, {\it i.e.}, downward, $\sin \varphi(t)\leq 0$. Regardless if it points downward or upward, the active rod's orientation ${\varphi}$ is governed by the same equation as in the bulk. Note that, as it is often done in applications (see, e.g., \cite{PotKaiBerAra2017}) we neglect inertia and translational diffusion in the individual based model.      

\begin{figure}[t]
	\centering
	\includegraphics[width=0.49\textwidth]{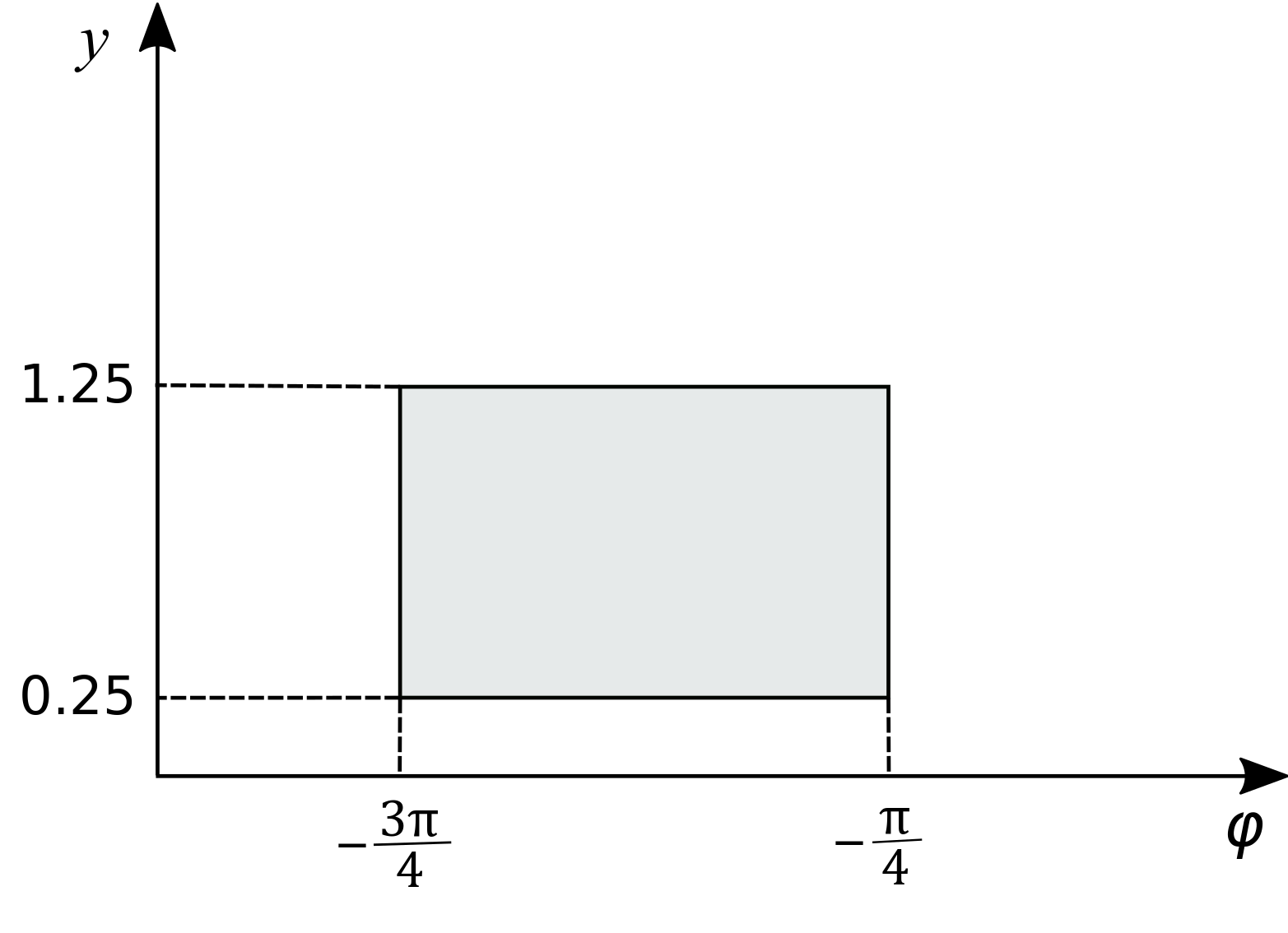}
	\includegraphics[width=0.49\textwidth]{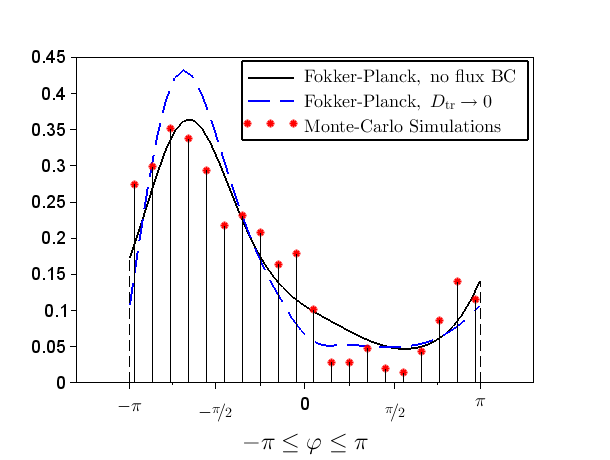}
	
	\medskip 
	\includegraphics[width=0.95\textwidth]{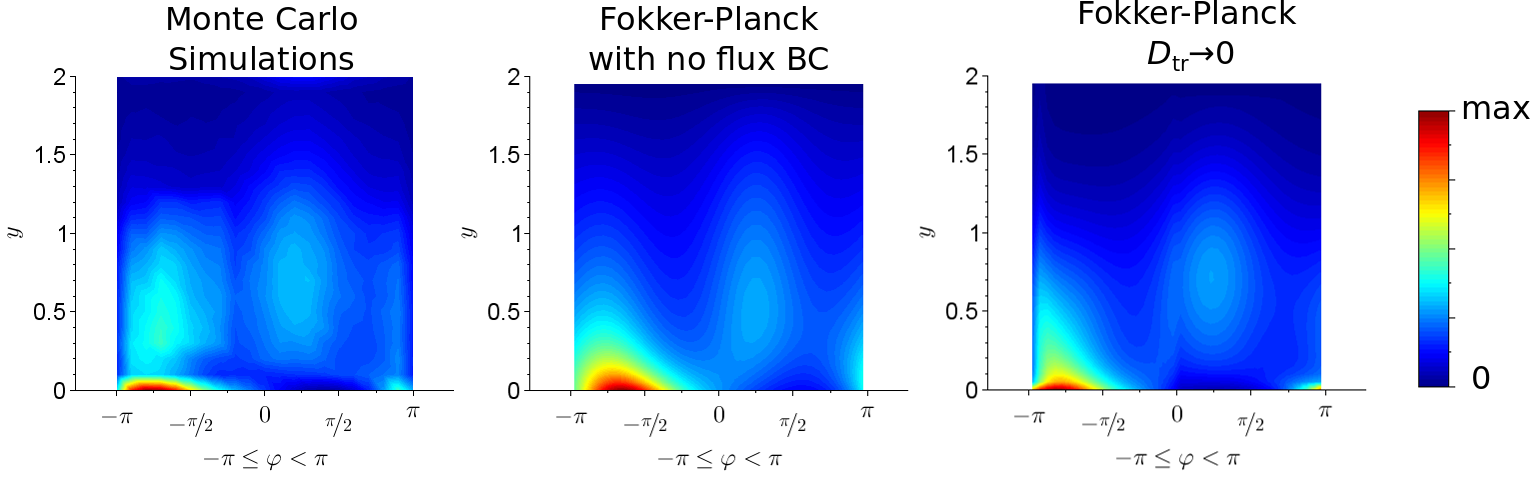}
	
	\caption{Upper Left: support $\{0.25<y<1.25,\,\sin \varphi<-\sqrt{2}/2\}$ of probability distribution function $\rho$ at $t=0$; Upper Right: angular distribution for $t=10$ of accumulated particles at wall by Monte-Carlo simulations (red dots), from Fokker-Planck problem \eqref{fp_num}-\eqref{fp_bc_num} (solid line), and two-scale expansions \eqref{twoscale0} (dashed line);  Lower figures:$(\varphi,y)$-histogram obtained from simulations of \eqref{ibm1}-\eqref{ibm2} (left); $(\varphi,y)$-distribution obtained from solution of \eqref{fp_num}-\eqref{fp_bc_num} (center);   $(\varphi,y)$-distribution obtained from two-scale expansions \eqref{twoscale0}, \eqref{def_of_rho_and_psi} with $\delta=0.05$ and $\rho_{\text{bulk}}$, $\psi_{\text{wall}}$ solving \eqref{eqn:no-diff} (right). All the lower plots are computed for~$t=10$.}%
	\label{fokker_planck_sim}
\end{figure}

Next, we compare results of Monte Carlo simulations for individual based model \eqref{ibm1}-\eqref{ibm2} with the initial boundary value problem  derived in Theorem~\ref{thm:main_mr} from Section~\ref{sec:zero-inertia}, consisting of the Fokker-Planck equation for probability distribution function $\rho(t,y,\varphi)$: 
\begin{equation}\label{fp_num}
\partial_t \rho + v_{\text{prop}}\sin \varphi \partial_y \rho - 0.5\dot{\gamma} \partial_{\varphi}((1-\cos(2\varphi))\rho)=D_{\text{tr}}\partial_{y}^2\rho+D_{\text{rot}}\partial_{\varphi}^2\rho,  
\end{equation}
and no-flux boundary condition: 
\begin{equation}\label{fp_bc_num}
D_{\text{tr}}\partial_y \rho=v_{\text{prop}} \rho\sin \varphi \quad \text{for} \quad y=0.
\end{equation}
Translation diffusion coefficient is chosen to be small, $D_{\text{tr}}=0.05$.

Finally, we simulate \eqref{eqn:no-diff}, derived as the limit $D_{\text{tr}}\to 0$ in \eqref{fp_num}-\eqref{fp_bc_num}.

A sample trajectory obtained from simulating \eqref{ibm1}-\eqref{ibm2} for $0<t<20$ is depicted in Fig.~\ref{fig:sketch}, right. This trajectory is drawn in $\varphi y$-plane, and rod locations within this plane at integer moment of times are marked by red dots while the value of the corresponding moment of time is written above each dot. The trajectory demonstrates typical behavior of an active rod swimming at a wall. After collision with the wall, ($t=4$), the rod attaches to the wall (it still can swim in $x$ direction) and re-orients under background shear decreasing orientation angle $\varphi$ to $\varphi = - \pi$, and then detaches from the wall ($t\approx 5.5$). Swimming with orientation $\varphi$ close to $\pm \pi$ with background flow given by $\boldsymbol{u}_{\text{BG}}=(\dot\gamma y, 0)$ means that the active rod swims upstream, that is, exhibits {\it negative rheotaxis}. 

For Monte Carlo simulations of \eqref{ibm1}-\eqref{ibm2} we chose initial location $y$ and orientation angle $\varphi$ to be random with uniform distributions in intervals $0.25<y<1.25$ and $-3\pi/4<\varphi<-\pi/4$, respectively. The corresponding initial condition for both probability distribution functions $\rho^{\text{I}}$, solution of Fokker-Planck equation \eqref{fp_num} with no flux boundary condition \eqref{fp_bc_num}, and $\rho^{\text{II}}$, given by two-scale expansion \eqref{twoscale0}, \eqref{def_of_rho_and_psi} with $\delta=D_{\text{tr}}=0.05$ and terms $\rho_{\text{bulk}}$ and $\psi_{\text{wall}}$ solving system \eqref{eqn:no-diff}, are 
\begin{equation}\label{initial_condition}
\left\{\begin{array}{rl}\dfrac{2}{\pi},&\quad \quad\dfrac{1}{4}<y<\dfrac{5}{4},\,-\dfrac{3\pi}{4}<\varphi<-\dfrac{\pi}{4},\\&\\0,&\quad \quad\text{otherwise}.\end{array}\right.
\end{equation} 
The initial condition \eqref{initial_condition} is shown in Fig.~\ref{fokker_planck_sim}, upper left. 

\begin{figure}[t]
	\begin{center}
		\includegraphics[width=0.325\textwidth]{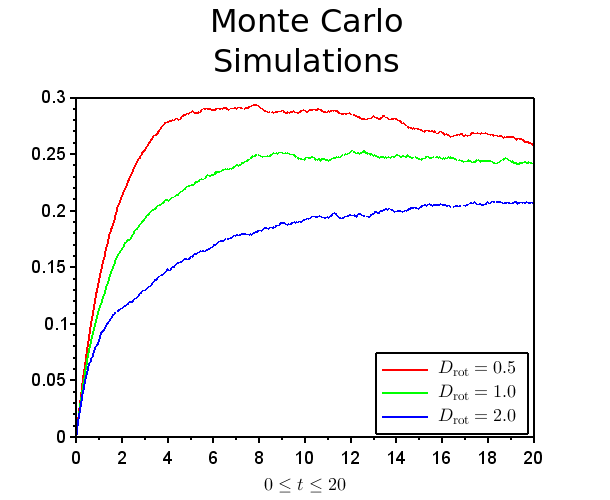}
		\includegraphics[width=0.325\textwidth]{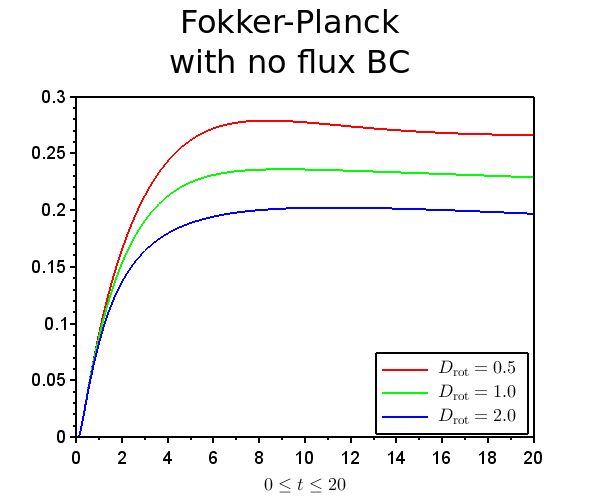}
		\includegraphics[width=0.325\textwidth]{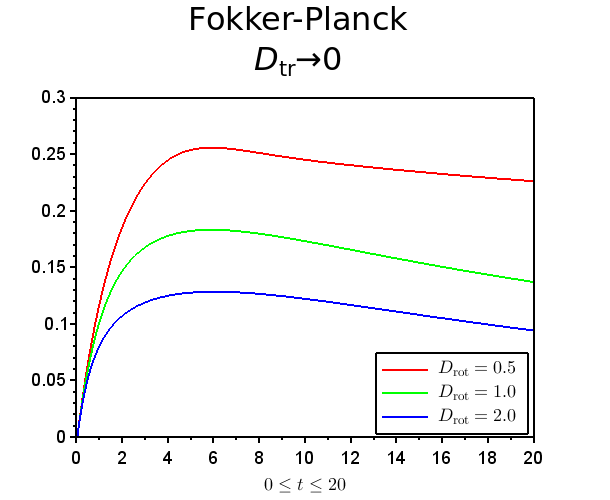}\\
		
		\caption{ Probability that active rod is accumulated at wall for $0<t<20$ from Monte-Carlo simulations (left), Fokker-Planck boundary value problem \eqref{fp_num}-\eqref{fp_bc_num} (center), and two-scale expansion \eqref{twoscale0}, \eqref{def_of_rho_and_psi} with $\delta=0.05$ and $\rho_{\text{bulk}}$, $\psi_{\text{wall}}$ solving \eqref{eqn:no-diff} (right).  }
		\label{fig:hist}
	\end{center}
\end{figure}

Results of numerical simulations are depicted in Fig.~\ref{fokker_planck_sim} and \ref{fig:hist}. Behavior of solutions of all the three problems $-$ Monte Carlo simulations of individual based model \eqref{ibm1}-\eqref{ibm2}, Fokker-Planck equation \eqref{fp_num} with no flux boundary condition \eqref{fp_bc_num}, and  two-scale expansion \eqref{twoscale0}, \eqref{def_of_rho_and_psi} with $\delta=D_{\text{tr}}=0.05$ and terms $\rho_{\text{bulk}}$ and $\psi_{\text{wall}}$ solving system \eqref{eqn:no-diff} $-$ is qualitatively similar. Distributions of location $y$ and orientation $\varphi$ concentrate at $y\approx 0$ and $\varphi\approx -\pi$ (Fig.~\ref{fokker_planck_sim}, lower row). These plots illustrate wall accumulation $y\approx 0$ and negative rheotaxis $\varphi\approx -\pi$. In this numerical example, initial condition \eqref{initial_condition} is chosen such that active rods initially point towards the wall, so the wall accumulation is somewhat enforced. Nevertheless, if domain is confined from all sides, which is a generic situation, an active rod reaches the wall with a high probability and spends a non-zero time at the wall reorienting before been detached. Therefore, the wall accumulation of active rods in confined bounded domains necessarily occurs, and the numerical example in this section investigates the situation when a population of active rods approach a wall. 

We also analyzed the active rod distribution inside boundary layer $\mathcal{L}:=\{0~\leq~y~<~0.2\}$ (accumulated active rods). Comparison of angular distributions at $t=10$, obtained from the three approaches, is given in Fig.~\ref{fokker_planck_sim}, upper right, and the probability of swimming inside boundary layer $\mathcal{L}$, as a function of time $0<t<20$, is depicted in Fig.~\ref{fig:hist}. All the three methods show active rod accumulation increases up to a moment $3<t<5$ (depending on the value of rotational diffusion coefficient $D_{\text{rot}}$; see Fig.~\ref{fig:hist}) and a peak of angular orientations forms close to $\phi = - \pi$ (see Fig.~\ref{fokker_planck_sim}, upper right). However, the third method slightly underestimates probability of swimming inside the boundary layer $\mathcal{L}$ for larger values of $D_{\text{rot}}$ (specifically, for $D_{\text{rot}}=2.0$). This underlines the subtlety, described in the paragraph after \eqref{rho_bulk}, of the multi-scale expansion for $\rho$ with respect to $\delta$ and presence of the scales additional to $1$ and $\delta^{-1}$ in \eqref{twoscale}. Rigorous analysis of the limit $\delta\to 0$ of solution to problem \eqref{FP_no_diff}-\eqref{no_flux_bc_b_layer} as well as explicit convergence estimates are left for our future work.


\section{Individual based model and Fokker-Planck equation for an active rod with inertia}
\label{sec:concrete}


\subsection{Individual based model for an active rod with inertia}
\label{sec:individual-rod}

In this section we present equations governing the motion of an active rod in the domain $\Omega\subset \mathbb R^2$ with non-empty boundary $\Gamma=\partial \Omega$. The rod is assumed to be a non-deformable one-dimensional segment of length $\ell$ swimming in a viscous fluid.  At each moment of time $t$ the rod is characterized by the location of its center of mass $\boldsymbol{r}(t)$ and the orientation angle $\varphi(t)$, so that the unit vector $\boldsymbol{p}(\varphi)=(\cos \varphi,\sin \varphi)$ with $\varphi=\varphi(t)$ determines the orientation of the active rod. Equations for $\boldsymbol{r}(t)$ and $\varphi(t)$ are 
\begin{eqnarray}
m \,\ddot{\boldsymbol{r}}&=&\eta_1 \ell \,(\boldsymbol{u}_{\text{BG}}(\boldsymbol{r})-\dot{\boldsymbol{r}})+F_{\mathrm{thrust}}\boldsymbol{p}(\varphi)+\sqrt{2D_{\mathrm{tr}}}\, {\zeta}_{1},\label{force_balance}\\
I_{\mathrm{rod}} \ddot{\varphi}&=&\eta_2 \ell^2(\mathrm{\Phi}_{\mathrm{BG}}(\boldsymbol{r},\varphi)- \dot{\varphi})+\sqrt{2D_{\mathrm{rot}}}\,{\zeta}_{2}.\label{torque_balance}
\end{eqnarray}  
Here $\rho_0$ is the density of the rod, so $m=\rho_0\ell$ is its mass; $I_{\mathrm{rod}}=\rho_0\ell^3/12$ is the moment of inertia of the  rod around the center of mass; $\eta_1$ and $\eta_2$ are material constants related to the background fluid viscosity. 
Function $\boldsymbol{u}_{\text{BG}}(\boldsymbol{r})$ is the velocity of the background flow at point $\boldsymbol{r}$. We assume that the background flow is not affected by the active rod thus $\boldsymbol{u}_{\text{BG}}$ satisfies the homogeneous Stokes equation in $\Omega$: 
\begin{eqnarray*}
	&-\eta_1\Delta \boldsymbol{u}_{\text{BG}}+\nabla p = 0,& \\
	&\nabla \cdot \boldsymbol{u}_{\text{BG}} = 0. &
\end{eqnarray*}
$F_{\mathrm{thrust}}$ is the magnitude of the thrust force (the self-propulsion force); $D_{\text{tr}}$ and $D_{\text{rot}}$ are translational and rotational diffusion coefficients, respectively; $\zeta_1$ and $\zeta_2$ are uncorrelated white noises with intensities $\langle\zeta_i(t),\zeta_i(t')\rangle=\delta(t-t')$, $i=1,2$. Function $\Phi_{\mathrm{BG}}(\boldsymbol{r},\varphi)$ can be understood as the proportionality coefficient defining the torque due to the background flow exerted on the unit rod at location ${\bf r}$, with orientation angle $\varphi$ and zero angular velocity, and $\Phi_{\mathrm{BG}}(\boldsymbol{r},\varphi)$ is given by 
\begin{equation}
\mathrm{\Phi}_{\mathrm{BG}}(\boldsymbol{r},\varphi)=({\mathbb I}-\boldsymbol{p}{\boldsymbol{p}}^{\mathrm{T}})\nabla \boldsymbol{u}_{\mathrm{BG}}(\boldsymbol{r}){\boldsymbol{p}}\cdot \mathrm{e}_{\varphi},\,\,\,\,\,\,\mathrm{e}_{\varphi}=(-\sin\varphi,\cos \varphi).\label{expression_for_torque}
\end{equation}

Equation \eqref{force_balance} is the force balance for the active rod, and this equation reads as follows. There are three forces exerted on the active rod: ({\it i}) the viscous drag force which is proportional to the relative velocity of the active rod with respect to the background flow, ({\it ii}) the self-propulsion force directed always along the direction of the active rod, and ({\it iii}) the random (Brownian) force. Equation \eqref{torque_balance} is the torque balance for the active rod which implies that there are two torques re-orienting the active rod: viscous and random torques. Here we assume that the self-propulsion force does not affect orientation of the active rod.     

The system \eqref{force_balance}-\eqref{torque_balance} can be written in the following general form 
\begin{eqnarray} 
\dfrac{\text{d}{\boldsymbol{r}}}{\text{d}t}&=& \boldsymbol{v}, \label{eq_r}\\ 
\dfrac{\text{d}\boldsymbol{v}}{\text{d}t}&=&\dfrac{1}{\varepsilon}(\boldsymbol{u}(\boldsymbol{r},\varphi,t) -  \boldsymbol{v})+\sqrt{2D_{\text{tr}}}\,\zeta_1,\label{eq_v} \\ 
\dfrac{\text{d}\varphi}{\text{d}t}&=&\omega, \label{eq_phi}\\ 
\dfrac{\text{d}\omega}{\text{d}t}&=&\dfrac{1}{\varepsilon}(T(\boldsymbol{r},\varphi,t)-\omega)+\sqrt{2D_{\text{rot}}}\,{\zeta}_2.  \label{eq_omega}
\end{eqnarray}
Here $u(\boldsymbol{r},\varphi,t)$ and $T(\boldsymbol{ r},\varphi,t)$ are given smooth and bounded functions and $\varepsilon$ is a small parameter. 


\begin{wrapfigure}[13]{r}{0.5\textwidth}
	 \vspace{-0.2 in}
	\fbox{
		\begin{minipage}{0.45\textwidth}
			\begin{center}
				\includegraphics[width=0.85\textwidth]{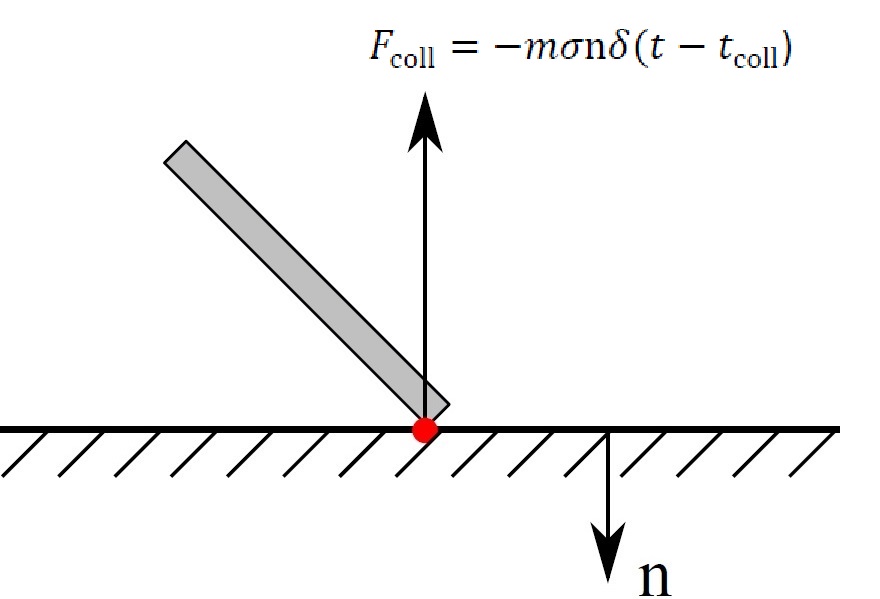}
				\caption{A force exerted on the active rod due to a collision with the wall.}
				\label{fig:collision}
			\end{center}
	\end{minipage}}
\end{wrapfigure}

Next, we describe the rule of collision of the active rod and the wall $\Gamma$. At time of a collision, $t_{\text{coll}}$, an instantaneous force (an impulse)  is exerted on the rod and this force is directed in direction normal to $\Gamma$ (see Fig.~\ref{fig:collision}):
\begin{equation*}
F_{\text{coll}}=-m\sigma \boldsymbol{n}\,\delta(t-t_{\text{coll}}),
\end{equation*}  
where $\boldsymbol{n}$ is the outward normal vector and $\sigma$ will be determined below. One can formally add the force $F_{\text{coll}}$ to the right hand side of equation \eqref{force_balance}. This translates into the following collision rule for velocities 
\begin{equation}
\label{collision_rule_velocities}
\boldsymbol{v}'=\boldsymbol{v}-\sigma \boldsymbol{n} 
\end{equation}
where $\boldsymbol{v}=\dot{\boldsymbol{r}}(t_{\mathrm{coll}}-0)$ and $\boldsymbol{v}'=\dot{\boldsymbol{r}}(t_{\mathrm{coll}}+0)$ are velocities before and after the collision, respectively. The force $F_{\mathrm{coll}}$ is exerted at the front of the active rod which touches the wall $\Gamma$. This implies that the active rod has an additional torque (in the right hand side of \eqref{torque_balance}) due to the collision, which equals to $-m\sigma \frac{\ell}{2}\,(\boldsymbol{p}\times \boldsymbol{n})\delta(t-t_{\text{coll}})$ and thus, the collision rule for angular velocities is 
\begin{equation}
\label{collision_rule_angular_velocities}
\omega'=\omega - \dfrac{6}{\ell}\sigma (\boldsymbol{p}\times \boldsymbol{n})\cdot \mathrm{e}_{z},
\end{equation} 
where $\omega=\dot{\varphi}(t_{\mathrm{coll}}-0)$ and $\omega'=\dot{\varphi}(t_{\mathrm{coll}}+0)$ are angular velocities before and after the collision, respectively, and $\text{e}_{z}=(0,0,1)$ is the unit vector orthogonal to the plane containing domain~$\Omega$.

To completely describe the collision rule, the value of parameter $\sigma$ needs to be determined. To this end, we assume that the collision is perfectly elastic, that is, the kinetic energy does not change in time of the collision:
\begin{eqnarray}\label{conservation_of_kinetic_energy}
&&\dfrac{1}{2}m \boldsymbol{v}^2+\dfrac{1}{2}I_{\mathrm{rod}}\omega^2=\dfrac{m}{2}(\boldsymbol{v}-\sigma\boldsymbol{n})^2+\dfrac{I_{\mathrm{rod}}}{2}(\omega - \dfrac{6}{\ell} \sigma \, (\boldsymbol{p}\times\boldsymbol{n})\cdot \mathrm{e}_z)^2. 
\end{eqnarray}
Expanding the right hand side of \eqref{conservation_of_kinetic_energy} we get the formula for $\sigma$:
\begin{equation}\label{def_of_sigma}
\sigma=\dfrac{2(\boldsymbol{v}\cdot \boldsymbol{n})+\ell\omega (\boldsymbol{p}\times\boldsymbol{n})\cdot \mathrm{e}_z}{1+3|\boldsymbol{p}\times\boldsymbol{n}|^2}.
\end{equation}   

\begin{remark}
	We can extend our consideration to imperfect elastic conditions, that is, we may assume that a part of energy is lost in time of collision: 
	\begin{equation}
	\dfrac{1}{2}m\boldsymbol{v}'^2+\dfrac{1}{2}I_{\mathrm{rod}}\omega'^2=\epsilon\left\{\dfrac{1}{2}m\boldsymbol{v}^2+\dfrac{1}{2}I_{\mathrm{rod}}\omega^2\right\},
	\end{equation}
	where $0\leq \epsilon\leq 1$ is coefficient of restitution which measures how elastic the collision is: if it is $1$, the collision is perfectly elastic; if  $\epsilon=0$, the collision is perfectly inelastic.
\end{remark}
\begin{proposition}\label{prop1}
	Let $\mathcal{C}$ be the linear operator (a $3\times 3$ matrix)\begin{equation}\label{def_of_C}
	\mathcal{C}:(\boldsymbol{v},\omega)\mapsto(\boldsymbol{v}-\sigma\, \boldsymbol{n},\omega-\dfrac{6}{\ell}\sigma (\boldsymbol{p}\times \boldsymbol{n})\cdot \mathrm{e}_z)
	\end{equation} 
	with $\sigma$ from \eqref{def_of_sigma}. Then
	\begin{equation}
	\label{substitution}
	|\mathrm{det}\,\mathcal{C}|=1.
	\end{equation}
\end{proposition}  
\begin{proof}
	Denote $v_{\mathrm{n}}:=\boldsymbol{v}\cdot \boldsymbol{n}$ and $v_{\tau}:=\boldsymbol{v}\cdot \boldsymbol{\tau}$ ($\boldsymbol{\tau}$ is the unit tangent vector on $\Gamma$). We also represent $\sigma$ as follows:
	\begin{equation*}
	\sigma = A v_{\mathrm{n}}+ B \omega, \text{ where } A=\frac{2}{1+3|\boldsymbol{p}\times \boldsymbol{n}|^2}\text{ and }B= \frac{\ell (\boldsymbol{p}\times\boldsymbol{n})\cdot \mathrm{e}_z}{1+3|\boldsymbol{p}\times \boldsymbol{n}|^2}.
	\end{equation*}
	In addition, we introduce angle $\theta$ between vectors $\boldsymbol{p}$ and $\boldsymbol{n}$. Note that $\sin \theta = (\boldsymbol{p}\times\boldsymbol{n})\cdot \mathrm{e}_z$. 
	
	\noindent Then the operator $\mathcal{C}$ can be represented by
	\begin{equation*}
	\mathcal{C}\mathcal{V}=\left[\begin{array}{rrr}1&0&0\\ 0& 1-A&-B\\0&-\dfrac{6}{\ell}\sin \theta\, A &1-\dfrac{6}{\ell}\sin \theta \,B\,\end{array}\right]\left[\begin{array}{c}v_{\tau}\\v_{\mathrm{n}}\\ \omega\end{array}\right].
	\end{equation*}
	
	\noindent Finally, we compute $\mathrm{det}\,\mathcal{C}$: 
	\begin{eqnarray*}
		\mathrm{det}\,\mathcal{C}&=& \left|\begin{array}{rr} 1-A&-B\\-\dfrac{6}{\ell}\sin \theta\, A &1-\dfrac{6}{\ell}\sin \theta\, B\end{array}\ \right|=-1.
	\end{eqnarray*}
	Thus, the proof of proposition is complete.
\end{proof}

\subsection{Fokker-Planck equation for an active rod with inertia}\label{subsec:fp}
For active rods whose motion inside domain $\Omega$ is described by \eqref{eq_r}-\eqref{eq_omega}, the Fokker-Planck equation is 
\begin{eqnarray}\label{fokkerplanck-0}
&&\partial_t \tilde f_{\ve}+\boldsymbol{v}\cdot \nabla_{\boldsymbol{r}} \tilde f_\ve+\dfrac{1}{\ve}\nabla_{\boldsymbol{v}}\cdot\left((\boldsymbol{u}-\boldsymbol{v}) \tilde f_\ve-\ve D_{\mathrm{tr}}\nabla_{\boldsymbol{v}} \tilde f_\ve\right)+ \nonumber
\\
&&\hspace{86pt}+\omega \partial_{\varphi} \tilde f_\ve+\dfrac{1}{\ve}\partial_{\omega}\left(( T-\omega) \tilde f_\ve-\ve D_{\mathrm{rot}}\partial_\omega  \tilde f_\ve\right)=0. 
\end{eqnarray}
The unknown function $ \tilde f_\ve(t,\boldsymbol{r},\boldsymbol{v},\varphi,\omega)$ is the probability distribution function of the active rod's location $\boldsymbol{r}\in \Omega$, translational velocity $\boldsymbol{v}\in \mathbb R^2$, orientation angle $\varphi \in [-\pi,\pi)$, and angular velocity $\omega\in \mathbb R$. 

The collision rule of the active rods with the wall $\Gamma$ is given by \eqref{collision_rule_velocities}-\eqref{collision_rule_angular_velocities}. The rule translates into the following boundary conditions for $ \tilde f_\ve$: 
\begin{equation}\label{bc_fokker-planck}
\boldsymbol{v} \tilde f_\ve\cdot \boldsymbol{n} =\boldsymbol{v}' \tilde f_\ve' \cdot (-\boldsymbol{n}),\,\,\,\boldsymbol{r}\text{ on }\Gamma,
\end{equation} 
where $ \tilde f_\ve'= \tilde f_\ve(t,\boldsymbol{r},\boldsymbol{v}',\varphi,\omega')$ and the pair $(\boldsymbol{v}',\omega')$ is given by collision rule \eqref{collision_rule_velocities}-\eqref{collision_rule_angular_velocities}.

\begin{remark}
	The meaning of boundary condition \eqref{bc_fokker-planck} is as follows: the flux of incident active rods equal to the flux of reflected active rods. In other words, the flux is the same before and after collisions. Note that if one considers spherical particles with no preferred orientation (no $\varphi$ and $\omega$), then the collision rule is $\boldsymbol{v}'=\boldsymbol{v}-2(\boldsymbol{v}\cdot \boldsymbol{n})\boldsymbol{n}$ and \eqref{bc_fokker-planck} is reduced in this case to equality of probability distribution functions $f_\ve=f'_\ve$. This boundary condition is used in many works where particles are assumed to be spherical, see e.g. \cite{GouJabVas2004,MelVas2007,MelVas2008}. We point out here that imposing such boundary conditions (equality of probability distribution functions instead of fluxes as in \eqref{bc_fokker-planck}) for active rods leads to violation of both the mass conservation and the energy relation.    	
\end{remark}

In this work we are interested in the limit $\ve \to 0$. Following \cite{GouJabVas2004,MelVas2007,MelVas2008}, to obtain a meaningful limit we first rescale $\tilde f_\ve$:
\begin{equation*}
f_\ve(t,\boldsymbol{r},\boldsymbol{v},\varphi,\omega)=\tilde f_\ve(t, \boldsymbol{r},\ve \boldsymbol{v},\varphi,\ve \omega).
\end{equation*}
Then the Fokker-Planck equation for the rescaled probability distribution function $f_\ve$ has the following form: 
\begin{eqnarray}\label{fokkerplanck}
&&\partial_t f_{\ve}+\dfrac{1}{\ve}\boldsymbol{v}\cdot \nabla_{\boldsymbol{r}}f_\ve+\dfrac{1}{\ve^2}\nabla_{\boldsymbol{v}}\cdot\left((\ve \boldsymbol{u}-\boldsymbol{v})f_\ve-D_{\mathrm{tr}}\nabla_{\boldsymbol{v}}f_\ve\right)+ \nonumber
\\
&&\hspace{86pt}+\dfrac{1}{\ve}\omega \partial_{\varphi}f_\ve+\dfrac{1}{\ve^2}\partial_{\omega}\left((\ve T-\omega)f_\ve-D_{\mathrm{rot}}\partial_\omega f_\ve\right)=0
\end{eqnarray}
with boundary condition \eqref{bc_fokker-planck} for $f_\ve$.


\section*{Acknowledgment}
\addcontentsline{toc}{section}{Acknowledgement}
The work of MP, LB, ER were supported by NSF DMREF grant DMS-1628411. PEJ was partially supported by NSF Grant 1614537, and NSF Grant RNMS (Ki-Net)
1107444.


\bibliographystyle{ieeetr}
\bibliography{zero_inertia}

\end{document}